\documentclass{LMCS}
\usepackage{enumerate}
\usepackage{hyperref}

\usepackage[active]{srcltx}
\usepackage[latin1]{inputenc}
\usepackage{mathpartir}
\usepackage{latexsym}
\usepackage{amssymb}
\usepackage{amsmath}
\usepackage{macros}
\usepackage{color}
\usepackage{graphics}
\usepackage{version}
\usepackage{url}
\usepackage{soul}
\usepackage{framed}

\newenvironment{definition}{\begin{defi}}{\end{defi}}
\newenvironment{lemma}{\begin{lem}}{\end{lem}}
\newenvironment{theorem}{\begin{thm}}{\end{thm}}

\newenvironment{proposition}{\begin{prop}}{\end{prop}}

\usepackage{smacros}
\usepackage{todo}
\definecolor{shadecolor}{rgb}{0.9,0.9,0.9}

\def\doi{7 (1:11) 2011}
\lmcsheading%
{\doi}
{1--44}
{}
{}
{Dec.~30, 2009}
{Mar.~29, 2011}
{}

\begin{document}

\title[Psi-calculi: a framework for mobile processes with nominal data and logic]{Psi-calculi:\\ a framework for mobile processes\\ with nominal data and logic}
\author[J.~Bengtson]{Jesper Bengtson}
\address{Department of Information Technology, Uppsala University, Sweden}
\email{\{jesper.bengtson,magnus.johansson,joachim.parrow,bjorn.victor\}@it.uu.se}

\author[M.~Johansson]{Magnus Johansson}
\address{\vskip-6 pt}

\author[J.~Parrow]{Joachim Parrow}
\address{\vskip-6 pt}

\author[B.~Victor]{Bj{\"o}rn Victor}
\address{\vskip-6 pt}

\date{\today}


\newcommand{\ALERT}[1]{\textcolor{red}{#1}}
\begin{abstract}
The framework of psi-calculi extends the pi-calculus with nominal datatypes for
data structures and for logical assertions and conditions. These can be
transmitted between processes and their names can be statically scoped as in the
standard pi-calculus.
Psi-calculi can capture the same phenomena as other proposed extensions of the pi-calculus such as the applied pi-calculus, the spi-calculus, the fusion calculus, the concurrent constraint pi-calculus, and calculi with polyadic communication channels or pattern matching. Psi-calculi can be even more general, for example by allowing structured channels, higher-order formalisms such as the lambda calculus for data structures, and predicate logic for assertions.

We provide ample comparisons to related calculi and discuss a few significant applications. Our labelled operational semantics and definition of bisimulation is straightforward, without a  structural congruence. We establish minimal requirements on the nominal data and logic in order to prove general algebraic properties of psi-calculi, all of which have been checked in the interactive theorem prover Isabelle.
Expressiveness of psi-calculi significantly exceeds that of other formalisms, while the purity of
the semantics is on par with the original pi-calculus.
\end{abstract}
\maketitle

\section{Introduction}
\label{sec:introduction}

The pi-calculus~\cite{milner.parrow.ea:calculus-mobile} has a multitude of extensions where  
higher-level data structures and operations on them are given as
primitive. To mention only two there are the
spi-calculus by Abadi and Gordon~\cite{abadi.gordon:calculus-cryptographic} focusing on cryptographic primitives, and the applied pi-calculus of Abadi and Fournet~\cite{abadi.fournet:mobile-values} where agents can introduce statically scoped  aliases of names for data, used e.g. to express how knowledge of an encryption is restricted.
It is also parametrised by an arbitrary signature for expressing data and an
equation system for expressing data equalities. The impact of these enriched
calculi is considerable with hundreds of papers applying or developing the
formalisms. As Abadi and Fournet rightly observe there is a trade-off between
``purity", meaning the simplicity and elegance of the original pi-calculus, and
modelling convenience; expressing complicated schemes in the original
pi-calculus can simply become too gruesome and error prone.

But the modelling convenience of many high-level primitives comes at a price. The theory of the formalism may instead become gruesome and error prone, and it can be difficult to assess the effects of modifications to it. 

Our contribution in this paper is to define psi-calculi: a framework where a range of calculi
can be formulated with a lean and symmetric semantics, and where proofs can be conducted using straightforward induction 
without the complications of stratified process definitions, 
structural congruence or explicit quantification of contexts.
We claim to be the first to formulate such truly compositional labelled
operational semantics for calculi of this calibre. Psi-calculi accommodate pi-calculus extensions such as
the spi-calculus, the applied pi-calculus, 
fusion~\cite{gardner.wischik:explicit-fusions},
concurrent constraints~\cite{buscemi.montanari:cc-pi},
and pi-calculus with polyadic synchronisation~\cite{carbone.maffeis:expressive-power}.

The main idea is that a psi-calculus is obtained by extending the basic untyped
pi-calculus with three parameters. The first is 
a set of data terms which can function as both
communication channels and communicated objects.
The second is a set of conditions, for use in conditional constructs such as {\bf if} statements. The third is
a set of assertions, used to express e.g. constraints or aliases, which can resolve the conditions. These sets need not be disjoint, and one of our main results is to identify minimal requirements on them.  They 
turn out to be quite general and natural.

Psi-calculi go beyond previous work on extending
pi-calculus since we allow
arbitrary assertions (and not only declarations of aliases), and
arbitrary conditions (and not only equality tests). Also, we base our exposition
on nominal datatypes and these accommodate e.g. alpha-equivalence classes of
terms with binders. For example, we can use a higher-order logic for assertions
and conditions, and higher-order formalisms such as the lambda calculus for data
terms and channels.
Thus we get the best of two worlds: expressiveness significantly exceeds that of the applied pi-calculus, while the ``purity'' of the semantics is on par with the original pi-calculus.

The straightforward definitions make our 
proofs
suitable for checking in a theorem prover. We have implemented our
framework in Isabelle~\cite{nipkow:isabelle} using its nominal data\-type
package~\cite{U07:NominalTechniquesInIsabelleHOL}, also known as Nominal Isabelle, and proved the algebraic properties of bisimilarity~\cite{BengtsonParrowTPHOLs2009}. This gives us absolute certainty of general results for a large class of calculi --- at least to the point of the current state of the art for machine checked proofs.

In the next section we give the basic definitions of the syntax and
semantics of psi-calculi. In Section~\ref{sec:expressiveness} we relate to other work
and demonstrate the expressiveness by showing how a variety of calculi
can be formulated. Section~\ref{sec:applications} contains more substantial examples on frequency hopping spread spectrum, multiple local services with a common global name, and cryptographic mechanisms including the Diffie-Hellman key agreement protocol. In Section~\ref{sec:bisimilarity} we introduce a
notion of bisimilarity, establish the expected algebraic results
about it, and demonstrate the proof of the most difficult parts. In Section~\ref{sect:formalisation} we discuss the full formalisation and implementation in Isabelle.
Finally Section~\ref{sec:further} concludes with ideas for further work.

This article extends \cite{DBLP:conf/lics/BengtsonJPV09}
by additional explanations, examples, and proofs, and a more strict
formalisation of some comparisons to related calculi. We are very grateful to
the three anonymous referees for many suggestions of improvements.

\section{Definitions}\label{sec:definitions}

\subsection{Nominal datatypes}
\label{sec:nominal}
We base psi-calculi on nominal datatypes. A reader unfamiliar with these need
not fear: we shall provide what little background is needed and be generous with
examples.  A traditional datatype can be built from a signature of constant
symbols, functions symbols, etc. A nominal datatype is more general, for
example it can also contain binders and identify alpha-variants of terms.
Formally a nominal datatype is not required to be built in any particular way;
the only requirements are related to the treatment of the atomic symbols called
names as explained below.

As usual we assume a countably infinite set of atomic {\em names} $\nameset$ ranged over by $a,\ldots,z$. Intuitively, names will represent the symbols that can be statically scoped, and also represent symbols acting as variables in the sense that they can be subjected to substitution. A typed calculus would distinguish names of different kinds but our account will be untyped. A typing may contribute to clarity of expressions but it is not necessary for our results.

A {\em nominal set}~\cite{PittsAM:nomlfo-jv,Gabbay01anew} is a set equipped with {\em name swapping} functions written $(a\;b)$, for any names $a,b$. An intuition is that for any member $X$ it holds that $(a\;b)\cdot X$ is $X$ with  $a$ replaced by $b$ and  $b$ replaced by $a$. Formally, a name swapping is any function  satisfying certain natural axioms such as \mbox{$(a\;b)\cdot ((a\;b)\cdot X) = X$}.
One main point of this is that  even though we have not defined any particular syntax we can define what it means for a name to ``occur" in an element: it is simply that it can be affected by swappings. 
The names occurring in this way  in an element $X$ constitute the  {\em support} of $X$, written $\n(X)$.
We write $a \freshin X$, pronounced ``$a$ is fresh for $X$", for $a \not\in
\n(X)$. In an inductively defined datatype without binders we will have
$a\freshin X$ if $a$ does not occur syntactically in $X$. In for example the
lambda calculus where alpha-equivalent terms are identified (i.e. the elements
are alpha-equivalence classes of terms) the support corresponds to the free
names. If $A$ is a set or a sequence of names we write $A\freshin X$ to mean
$\forall a \in A \;.\; a \freshin X$.

We require all elements to have finite support, i.e., $\n(X)$ is finite for all
$X$. It follows that for any $X$ there are infinitely many $a$ such that $a
\freshin X$. Some elements will have empty support, a prime example is the
identity function in the lambda calculus, or a term of a traditional datatype
not containing any names. 
A function $f$ is \emph{equivariant} if $(a\;b)\cdot f(X)
= f((a\;b)\cdot X)$ holds for all $X$, and similarly for functions and relations of any
arity. Intuitively, this means that all
names are treated equally.

A {\em nominal datatype} is a nominal set together with a set of equivariant
functions on it. In particular we shall consider substitution functions that
substitutes elements  for  names. If $X$ is an element of a datatype,
$\tilde{a}$ is a sequence of names without duplicates and $\tilde{Y}$ is an
equally long sequence of elements of possibly another datatype, the {\em
substitution}
$X\lsubst{\tilde{Y}}{\tilde{a}}$ is an element of the same datatype as $X$.
In a traditional datatype
substitution can be thought of as replacing all occurrences of names $\tilde{a}$
by $\tilde{Y}$. In a calculus with binders it can be thought of as replacing the
free names, alpha-converting any binders to avoid capture.

For the purpose of psi-calculi it turns out that we need not define exactly what a substitution does.
The only formal requirements are that substitution is an equivariant function that 
satisfies two substitution laws:

\bigskip
\begin{tabular}{@{}ll}
1: & if $\; \vec{a}\subseteq\n(X)$ and $b \in \n(\vec{T})$ then $b \in \n(X[\vec{a}:=\vec{T}])$

\\
2: &if $\vec{b} \freshin X, \vec{a}$ then $X[\vec{a}:=\vec{T}] = ((\vec{b}\;\vec{a})\cdot X)[\vec{b}:=\vec{T}]$
\end{tabular}
\bigskip

\noindent
Law 1 says that substitutions may not lose names: any name $b$ in the objects
$\vec{T}$ that substitute for names $\vec{a}$ occurring in $X$  must also appear
in the substitution $X[\vec{a}:=\vec{T}]$. Law~2 is a form of alpha-conversion
for substitutions; here it is implicit that $\vec{a}$ and $\vec{b}$ have the
same length, and $(\vec{a}\; \vec{b})$ swaps each element of $\vec{a}$ with the
corresponding  element of $\vec{b}$.  At the end of
Section~\ref{sec:transitions} we shall motivate why these laws are necessary.
 
 \paragraph{Example:}
Consider an inductively defined datatype without binders, where the support is
the set of names that occur syntactically, and substitution is the syntactic
replacement of names for terms, defined inductively in the usual way. The
arguments that this substitution function satisfies our requirements are
straightforward. Equivariance and Law 2 follow  immediately by induction. For
Law 1,
suppose $\vec{a} \subseteq \names{X}$. This means that all elements of $\vec{a}$ occur syntactically in $X$. Suppose $b \in \names{\vec{T}}$. This means that for some $i$, $b \in \names{T_i}$. This means that $b$ occurs syntactically in $T_i$. Consider the corresponding $a_i$. We know $a_i$ occurs syntactically in $X$. So then by definition $T_i$ occurs syntactically in $X[\vec{a}:=\vec{T}]$. Therefore $b$ occurs syntactically in that term, and by definition is in the support of it.

The main point of using nominal datatypes is that we obtain a general framework,
allowing many different instantiations.
Our only requirements are on the notions of support, name swapping, and substitution. This corresponds precisely to the essential ingredients for data transmitted between agents. Since names can be statically scoped and data sent into and out of scope boundaries, it must be possible to discern exactly what names are contained in what data items, and this is just the role of the support. In case a data element intrudes a scope, the scoped name needs to be alpha converted to avoid clashes, and name swapping can achieve precisely this. When a term is received in a communication between agents it must replace all occurrences of the placeholder in the input construct, in other words, the placeholder is substituted by the term. 

Since these are the only things we assume about data terms we can handle
datatypes that are not inductively defined, such as equivalences classes and
sets defined by comprehension or co-induction. Examples include higher-order
datatypes such as the lambda calculus. As long as it
satisfies the axioms of a nominal datatype
it can be used in our framework. Similarly, the notions of conditions, i.e., the
tests on data that agents can perform during their execution, and assertions,
i.e. the facts that can be used to resolve conditions, are formulated as nominal
datatypes. This means that logics with binders and even higher-order logics can
be used. Moreover, alpha-variants of terms can be formally equated by taking the
quotient of terms under alpha equality, thereby facilitating the formalism and
proofs.

\subsection{Terms, conditions, and assertions}

\newcommand{\terms}{{\rm\bf T}}
\newcommand{\conditions}{{\rm\bf C}}
\newcommand{\assertions}{{\rm\bf A}}

Formally, a psi-calculus is defined by instantiating three nominal datatypes and
four operators:
\begin{definition}[Psi-calculus parameters]
\label{def:parameters}
A psi-calculus requires the three (not necessarily disjoint) nominal datatypes:
\[\begin{array}{ll}
\terms & \mbox{the (data) terms, ranged over by $M,N$} \\
\conditions  & \mbox{the conditions, ranged over by $\varphi$}\\
\assertions & \mbox{the assertions, ranged over by $\Psi$}
\end{array}\]
and the four equivariant operators:
\[\begin{array}{ll}
\sch:  \terms \times \terms \to \conditions & \mbox{Channel Equivalence} \\
\ftimes: \assertions \times \assertions \to \assertions& \mbox{Composition} \\
\emptyframe: \assertions& \mbox{Unit} \\
\vdash\,\subseteq \assertions \times \conditions & \mbox{Entailment}
\end{array}
\]
and substitution functions $\lsubst{\ve{M}}{\ve{a}}$, substituting terms for names, on all of {\bf T}, {\bf C} and {\bf A}.
\end{definition}

As an example, we can choose data terms inductively generated by some signature, assertions and conditions to be elements of a first-order logic with equality over these terms, entailment to be logical implication, $\ftimes$ to be conjunction and $\emptyframe$ to be {\sc true}. 

The binary functions above will be written in infix. Thus, if $M$ and $N$ are terms then $M \sch N$ is a condition, pronounced ``$M$ and $N$ are channel equivalent" and if $\Psi$ and $\Psi'$ are assertions then so is $\Psi \ftimes \Psi'$. Also we write $\Psi \vdash \varphi$, pronounced ``$\Psi$ entails $\varphi$", for $(\Psi, \varphi) \in \;\vdash$.

The data terms are used to represent all kinds of data, including communication channels.
Intuitively, two agents can communicate if one sends and the other receives along the same channel.
This is why we require a condition
 \mbox{$M\sch N$}  to say that $M$ and $N$ represent the same communication channel.  For example, in the pi-calculus $\sch$ is just identity of names.

The assertions will be used to declare information necessary to resolve the conditions. Assertions can be contained in agents and represent constraints; they can contain names and thereby be syntactically scoped and represent information known only to the agents within that scope. 
The operator $\ftimes$ on assertions will, intuitively, be used to represent conjunction of the information in the assertions.  The assertion $\emptyframe$ is the unit for $\ftimes$.
The intuition of entailment is that $\Psi \vdash \varphi$ means that given the information in $\Psi$, it is possible to infer $\varphi$.
We say that two assertions are equivalent if they entail the same conditions:

\begin{definition}[assertion equivalence]
Two assertions are {\em equivalent}, written $\Psi \sequivalent \Psi'$, if for all $\varphi$ we have that $\Psi \vdash \varphi \Leftrightarrow\Psi' \vdash \varphi$.
\end{definition}

We can now formulate our requisites on valid psi-calculus parameters:

\begin{definition}[Requisites on valid psi-calculus parameters]
\label{def:entailmentrelation}
\
\begin{mathpar}
\begin{array}{ll}
\mbox{Channel Symmetry:} & \Psi \vdash M \sch N \quad \Longrightarrow\quad \Psi \vdash N \sch M \\
\mbox{Channel Transitivity:} & \Psi \vdash M \sch N \; \land \; \Psi \vdash N \sch L
 \quad \;\Longrightarrow\quad \Psi \vdash M \sch L\\

\mbox{Compositionality:} & \Psi \sequivalent \Psi'  \quad\Longrightarrow\quad \Psi \ftimes \Psi'' \sequivalent \Psi' \ftimes \Psi''\\
\mbox{Identity:} & \Psi \ftimes \emptyframe \;\sequivalent\; \Psi \\
\mbox{Associativity:}& (\Psi \ftimes \Psi') \ftimes \Psi'' \;\sequivalent\; \Psi \ftimes (\Psi' \ftimes \Psi'')\\
\mbox{Commutativity:}&  \Psi \ftimes \Psi' \;\sequivalent\; \Psi' \ftimes \Psi \\

\end{array}
\end{mathpar}
\end{definition}
\noindent
Our requisites on a psi-calculus are that the channel equivalence is a partial equivalence relation, that $\ftimes$ is compositional, and that the equivalence classes of assertions form an abelian monoid. In Section~\ref{sec:illustrativeexamples} below we will demonstrate that all requisites in Definition~\ref{def:entailmentrelation} are essential.

Note that channel equivalence is not required to be reflexive. Thus it is possible to have data terms that are not channel equivalent to anything at all, meaning that they cannot be used as channels. Also, note that properties such as weakening \mbox{($\Psi \vdash \varphi \Rightarrow \Psi \ftimes \Psi' \vdash \varphi$)} and idempotence ($\Psi \ftimes \Psi \sequivalent \Psi$) are not required. This means that we can in principle represent non-monotonic logics as well as logics to represent resource use, although these avenues remain yet unexplored. A main point of our work is to identify minimal requisites for the formal results on bisimilarity to hold, and here neither weakening nor idempotence turns out to be necessary.

\subsection{Frames}
\label{sec:frames}
Assertions can contain information about names, and names can be scoped using the familiar pi-calculus operator $\nu$. For example, in a cryptography application an assertion $\Psi$ could be that the a datum represents the encoding of a message using a key $k$. This $\Psi$ can occur under the scope of $\nu k$, to signify that the key is known only locally. In order to admit this in a general way we use the notion of a frame, first introduced by Abadi and 
Fournet~\cite{abadi.fournet:mobile-values}. Basically, a frame is just an assertion with additional information about which names are scoped. The example above where $\Psi$ occurs under the scope of $k$ will be written
$\framepair{k}{\Psi}$, to signify a frame consisting of the assertion $\Psi$ where the name $k$ is local.

In the following $\tilde{a}$ means a finite (possibly empty) sequence of names, $a_1,\ldots,a_n$. The empty sequence is written $\epsilon$ and the concatenation of $\tilde{a}$ and $\tilde{b}$ is written $\tilde{a} \tilde{b}$.
When occurring as an operand of a set operator, $\tilde{a}$ means the corresponding set of names $\{a_1,\ldots, a_n\}$. We also use sequences of terms, conditions, assertions etc. in the same way.

\begin{definition}[Frame]
\label{def:frame}
A {\em frame} is of the form   $\framepair{\frnames{}}{\frass{}}$ where
$\frnames{}$ is a sequence of names that bind into the assertion 
$\frass{}$. We identify alpha variants of frames.\footnote{In some presentations frames have been written just as pairs $\langle \frnames{},\frass{} \rangle$. The notation in this paper better conveys the idea that the names bind into the assertion, at the slight risk of confusing frames with agents. Formally, we establish frames and agents as separate types, although a valid intuition is to regard a frame as a special kind of agent, containing only scoping and assertions. This is the view taken in~\cite{abadi.fournet:mobile-values}. }

\end{definition}

We use $F,G$ to range over frames. Since we identify alpha variants we can always choose the bound names freely.

Notational conventions: We write just $\Psi$ instead of $(\nu\epsilon)\Psi$ when there is no risk of confusing a frame with an assertion, and $\ftimes$ to mean composition on frames
defined by $\framepair{\frnames{1}}{\frass{1}} \ftimes
\framepair{\frnames{2}}{\frass{2}} = 
\framepair{\frnames{1} \frnames{2}}{\frass{1} \ftimes \frass{2}}$ where
$\frnames{1}$ $\freshin$ $\frnames{2},\frass{2}$ and vice versa. We
write 
$(\nu c)(\framepair{\frnames{}}{\frass{}})$ to mean $\framepair{c\frnames{}}{\frass{}}$.

Intuitively a condition is entailed by a frame if it is entailed by the assertion and does not contain any names bound by the frame. 
Two frames are equivalent if they entail the same conditions:
\begin{definition}[Equivalence of frames]\label{def:frame-equivalence}
We define $F \vdash \varphi$ to mean that there exists an alpha variant $\framepair{\frnames{}}{\frass{}}$ of $F$ such that  $\frnames{}
\freshin \varphi$ and $\frass{} \vdash \varphi$. We also define 
$F\sequivalent G$ to mean that for all $\varphi$ it holds that $ F \vdash
\varphi$ iff $ G \vdash \varphi$.
\end{definition}
For example $(\nu a b)\Psi \sequivalent (\nu b a)\Psi$, and  if $a \freshin \Psi$ then $(\nu a)\Psi \sequivalent \Psi$.

\label{sec:first-crypto-example}
To take an example of first-order logic with equality, assume that the term $\sym{enc}{M,k}$ represents the encoding of message $M$ with key $k$.
Let $\Psi$ be the assertion $C = \sym{enc}{M,k}$, stating that the ciphertext $C$ is the result of encoding $M$ by $k$. If an agent contains this assertion the environment of the agent will be able to use it to resolve tests on the data, in particular to infer that $C = \sym{enc}{M,k}$. In other words, if the environment receives $C$ it can test if this is the encryption of $M$. In order to restrict access to the key $k$ it can be enclosed in a scope $\nu k$. 
 The environment of the agent will then have access to the frame $(\nu k)\Psi$ rather than $\Psi$ itself. This frame is much less informative, for example it does {\em not} hold that 
$(\nu k)\Psi \vdash C = \sym{enc}{M,k}$. Here great care has to be made to formulate the class of allowed conditions. If these only contain equivalence tests of terms, $(\nu k)\Psi$ will entail nothing but tautologies and be equivalent to $\emptyframe$. But if quantifiers are allowed in the conditions, then by existential introduction $\Psi \vdash \exists k. (C = \sym{enc}{M,k})$, and since this condition has no free $k$ we get $(\nu k)\Psi \vdash \exists k. (C = \sym{enc}{M,k})$. In other words the environment will learn that $C$ is the encryption of $M$ for some key $k$. We shall return to examples related to cryptography in Section~\ref{sect:api}.

Most of the properties of assertions carry over to frames. Channel symmetry and channel transitivity, identity, associativity and commutativity all hold, but  compositionality in general does not. In other words, there are psi-calculi with frames $F,G,H$ where
$F \sequivalent G$ but not $F \ftimes H \sequivalent G \ftimes H$. An example is if there are assertions $\Psi$, $\Psi'$ and $\Psi_a$ for all names $a$, conditions $\varphi'$ and $\varphi_a$ for all names $a$, and where the entailment relation satisfies $\Psi_a \vdash \varphi_a$ and $\Psi' \vdash \varphi'$. Suppose composition is defined such that $\Psi \ftimes \Psi = \Psi$ and all other compositions yield $\Psi'$. By adding a unit element this satisfies all requirements on a psi-calculus. In particular $\ftimes$ is trivially compositional because no two different assertions are equivalent. Also $(\nu a)\Psi_a \sequivalent \Psi$, but $\Psi \ftimes (\nu a)\Psi_a \not\sequivalent \Psi \ftimes \Psi$ since $\Psi \ftimes \Psi_a = \Psi' \vdash \varphi'$.

\subsection{Agents}
\label{sec:agents}
\begin{definition}[psi-calculus agents]\label{def:agents}
Given valid psi-calculus parameters as in Definitions~\ref{def:parameters} and~\ref{def:entailmentrelation}, the psi-calculus {\em agents}, ranged over by $P,Q,\ldots$,  are of the following forms.
{\rm
\[
\begin{array}{ll}

\nil                          & \mbox{Nil} \\
\out{M}N \sdot P                   & \mbox{Output} \\
\lin{M}{\ve{x}}N \sdot P          & \mbox{Input}\\
\caseonly{\ci{\varphi_1}{P_1}\casesep\cdots\casesep\ci{\varphi_n}{P_n}}
&\mbox{Case} \\
(\nu a)P                      & \mbox{Restriction}\\
P \pll Q                      & \mbox{Parallel}\\
! P                           & \mbox{Replication} \\
\pass{\Psi}                        & \mbox{Assertion}\\
\end{array}\]
}
\noindent
In the Input $\lin{M}{\ve{x}}N.P$ we require that $\ve{x} \subseteq \names{N}$ is a sequence without duplicates, and the names $\ve{x}$ bind occurrences in both $N$ and $P$.  Restriction binds $a$ in $P$. We identify alpha equivalent agents. An assertion is {\em guarded} if it is a subterm of an Input or Output. In a replication $!P$ there may be no unguarded assertions in $P$, and in $\caseonly{\ci{\varphi_1}{P_1}\casesep\cdots\casesep\ci{\varphi_n}{P_n}}$ there may be no unguarded assertion in any $P_i$.
\end{definition}
In the Output and Input forms $M$ is called the subject and $N$ the object.
Output and Input  are similar to those in the pi-calculus, but arbitrary terms can function as both subjects and objects. In the input $\lin{M}{\ve{x}}N.P$ the intuition is that the pattern $(\lambda \ve{x}) N$ can match any term obtained by instantiating $\ve{x}$,
e.g., $\lin{M}{x,y}f(x,y).P$ can only communicate with an output
$\out{M}{f(N_1,N_2)}$ for some data terms $N_1, N_2$.
This can be thought of as a generalisation of the polyadic pi-calculus where the patterns are just tuples of names.
Another significant extension  is that we allow arbitrary data terms
also as communication channels. Thus it is possible to include functions that create channels.

The {\bf case} construct as expected works by behaving as one of the $P_i$ for which the corresponding $\varphi_i$ is true. 
$\caseonly{\ci{\varphi_1}{P_1}\casesep\cdots\casesep\ci{\varphi_n}{P_n}}$
 is sometimes abbreviated as
\mbox{\rm $\caseonly{\ci{\ve{\varphi}}{\ve{P}}}$}, or if $n=1$ as 
$\ifthen{\varphi_1}{P_1}$.
In psi-calculi where a condition $\top$ exists such that $\Psi \vdash \top$ for all $\Psi$
 we write $P+Q$ to mean $\caseonly{\ci{\top}{P}\casesep\ci{\top}{Q}}$.
 
Input subjects are underlined to facilitate parsing of complicated expressions; in simple cases we often omit the underline.
In the traditional pi-calculus terms are just names and its input construct $a(x) \sdot P$ can be represented as $\lin{a}{x}x.P$. In some of the examples to follow we shall use the simpler notation  $a(x)\sdot P$ for this input form,
and sometimes we omit a trailing $\nil$, writing just $\out{M}N$ for  $\out{M}N \sdot \nil$. If the object of an Output is a long term we enclose it in brackets $\langle \; \rangle$ to make it easier to parse. 

For a simple example, the pi-calculus~\cite{milner.parrow.ea:calculus-mobile}
can be represented as a psi-calculus where the only data terms are names,  the
only assertion is $1$, and the conditions are equality tests on
names. Substitution is the standard capture-avoiding syntactic replacement of names for names.
We call this instance \piinstance, and formally we have:
\[
\begin{array}{rcl}
\terms & \defn & \nameset\\
\conditions & \defn &\{a=b : a,b\in \terms\}\\
\assertions & \defn &\{1\}\\
\sch & \defn & =\\
\ftimes & \defn &\lambda \Psi_1, \Psi_2.\; 1 \\
\unit & \defn & 1\\
\vdash & \defn & \{(1, a=a) : a \in \nameset\}\}
\end{array}
\]

We can represent pi-calculus choice using the \textbf{case} statement:
the pi-calculus term $P + Q$ corresponds to
$(\nu a)(\mbox{\bf case}\; a=a \;:\; P \;[]\; a=a \,:\, Q)$, where $a \freshin
P, Q$, and pi-calculus match $[a=b]P$ to $\mbox{\bf if}\; a=b\;
\mbox{\bf then}\; P$. We will return to this instance in
Section~\ref{sec:expressiveness}.

\label{sec:polyadic-pi-instance}
We obtain the polyadic pi-calculus  by adding the tupling symbols
${\mathrm{t}}_n$ for tuples of arity $n$ to $\terms$., i.e.
$\terms = \mathcal{N} \cup \{\mathrm{t}_n(M_1,\ldots,M_n):\;M_1,\ldots, M_n \in \terms\}$.  The polyadic output is to simply output the corresponding tuple of object names, and the polyadic input $a(b_1,\ldots,b_n)\sdot P$ is represented by a pattern matching
$\underline{a}(\lambda{b_1,\ldots,b_n})\sym{t_\mathit{n}}{b_1,\ldots,b_n}\sdot P$. Strictly speaking this allows nested tuples and tuples also in subject position in agents, but as we shall see such prefixes will not give rise to any transition, since in this psi-calculus $M \sch M$ is only entailed when $M$ is a name, i.e., only names are channels.

 In a psi-calculus the channels can be arbitrary terms. This means that it is    
 possible to
introduce functions on channels (e.g., if $M$ is a channel then so is $f(M)$). 
It also means that a channel can contain more than one name. An extension of
this kind is explored by Carbone and
Maffeis~\cite{carbone.maffeis:expressive-power} in the so called pi-calculus
with polyadic synchronisation, ${}^e\pi$. Here action subjects are tuples of names, and it is
demonstrated that this allows a gradual
enabling of communication by opening the scope of names in a subject,
results in simple representations of localities and cryptography, and  gives a
strictly greater expressiveness than standard \pic{}. 
We can represent ${}^e\pi$ by using tuples of names in subject position. The only modification to the representation of the polyadic pi-calculus is to extend $\vdash$ to  \mbox{$\vdash = \{(\emptyframe, M\sch M) : \; M \in \terms\}$}, and to remove the conditions of type $M=N$ (since they can be encoded in ${}^e\pi$).

The data terms can also be drawn from a higher-order formalisms. It is thus possible to transmit functions between agents. For example, let $\terms$ be the lambda calculus, containing abstractions 
$\mbox{\boldmath$\lambda$} x.M$ and applications $MN$. In the parallel composition
\(\overline{a}\,\langle \mbox{\boldmath$\lambda$} x.M\rangle \sdot P \parop a (z) \sdot \overline{b}\,\langle zN\rangle \sdot Q\)
the left hand component transmits the function $\mbox{\boldmath$\lambda$} x.M$ to the right, where the application of it to $N$ is transmitted along $b$. Reduction would be represented as a binary predicate over lambda terms and could be tested in psi-calculus conditions (the reduction rules would be part of the definition of entailment).
In this sense psi can resemble a higher-order calculus. It is even possible to let the terms be the psi-calculus agents themselves. An agent transmitted as a term cannot directly communicate with the agent that sent or received it, but there is a possibility of indirect interaction through the entailment relation.  This area we leave for further study.

\subsection{Operational semantics}
\label{sec:transitions}
In this section we define an inductive transition relation on agents.
In particular it establishes what transitions are possible from a parallel composition $P\pll Q$. In the standard pi-calculus the transitions from a parallel composition can be uniquely determined by the transitions from its components, but in psi-calculi the situation is more complex. Here the assertions contained in $P$ can affect the conditions tested in $Q$ and vice versa. For this reason we introduce the notion of the {\em frame of an agent} as the combination of its top level assertions,
retaining all the binders.
It is precisely this that can affect a parallel agent.

\begin{definition}[Frame of an agent]
The {\em frame $\fr{P}$ of an agent} P is defined inductively as follows:
\[\begin{array}{l}
\fr{\nil} = \fr{\lin{M}{\ve{x}}N.P} = \fr{\out{M}N.P} = 
\fr{\caseonly{\ci{\ve{\varphi}}{\ve{P}}}} =
\fr{!P} = \emptyframe \\
\fr{\pass{\Psi}} = \Psi \\
\fr{P \pll Q} = \fr{P}\ \ftimes\ \fr{Q}\\
\fr{\res{b}P} = (\nu b)\fr{P}  \\
 \end{array}\]
\end{definition}
\noindent
For a simple example, if $a \freshin \Psi_1$:
\[\fr{\pass{\Psi_1} \,|\,(\nu a)(\pass{\Psi_2} \,|\, \out{M}N.\pass{\Psi_3}}
\quad= \quad\framepair{a}{(\Psi_1 \ftimes \Psi_2)}\]
Here $\Psi_3$ occurs under a prefix and is therefore not included in the frame.
An agent where
all assertions are guarded thus has a frame equivalent to $\emptyframe$. In the following  we often write   $\framepair{\frnames{P}}{\frass{P}}$ for  $\fr{P}$, but note that this is not a unique representation since frames are identified up to alpha equivalence.

The actions $\alpha$ that agents can perform are of three kinds: output actions, input actions of the early kind, meaning that the input action contains the received object, and the silent action $\tau$. The operational semantics consists of transitions of the form
$\framedtransempty{\Psi}{P}{\alpha}{P'}$.
This transition intuitively means that $P$ can perform an action $\alpha$ leading to $P'$, in an environment that asserts $\Psi$.

\begin{definition}[Actions]

The {\em actions} ranged over by $\alpha, \beta$ are of the following three kinds:

\[\begin{array}{lllll}

\out{M}{(\nu \tilde{a})N} & \mbox{Output, where $\tilde{a} \subseteq \n(N)$} \\
\inn{M}N                  & \mbox{Input}  \\
\tau                      & \mbox{Silent}
\end{array}\]

\end{definition}
For actions we refer to $M$ as the {\em subject} and $N$ as the {\em object}. We define 
$\bn{\out{M}{(\nu \tilde{a})N}} = \tilde{a}$, and $\bn{\alpha}=\emptyset$ if $\alpha$ is an input or $\tau$. We also define $\n(\tau)=\emptyset$ and $\n(\alpha) = \n(N) \cup \n(M)$ if $\alpha$ is an output or input.
As in the pi-calculus, the output $\out{M}{(\nu \tilde{a})N}$ represents an action sending $N$ along $M$ and opening the scopes of the names $\tilde{a}$.  Note in particular that the support of this action includes  $\tilde{a}$. Thus  $\out{M}{(\nu a)a}$ and $\out{M}{(\nu b)b}$ are different actions.

\begin{table*}[tb]

\begin{mathpar}

\inferrule*[Left=\textsc{In}]
    {\Psi \vdash M \sch K }
    {\framedtransempty{\Psi}{\lin{M}{\ve{y}}{N}.P}{\inn{K}{N}\lsubst{\ve{L}}{\ve{y}}}{P\lsubst{\ve{L}}{\ve{y}}}}

\inferrule*[left=\textsc{Out}]
    {\Psi \vdash M \sch K }
    {\framedtransempty{\Psi}{\out{M}{N}.P}{\out{K}{N}}{P}}

\inferrule*[left={\textsc{Case}}]
    {\framedtransempty{\Psi}{P_i}{\alpha}{P'} \\ \Psi \vdash \varphi_i}
    {\framedtransempty{\Psi}{\caseonly{\ci{\ve{\varphi}}{\ve{P}}}}{\alpha}{P'}}

\inferrule*[left=\textsc{Com}, Right={$\inferrule{}{\ve{a} \freshin Q }$}]
 {\framedtransempty{\frass{Q} \ftimes \Psi}{P}{\bout{\ve{a}}{M}{N}}{P'} \\
  \framedtransempty{\frass{P} \ftimes \Psi}{Q}{\inn{K}{N}}{Q'} \\
  \Psi \ftimes \frass{P} \ftimes \frass{Q} \vdash M \sch K
  }
       {\framedtransempty{\Psi}{P \pll Q}{\tau}{(\nu \ve{a})(P' \pll Q')}}

\inferrule*[left=\textsc{Par},  right={$\bn{\alpha} \freshin Q$
}]
{\framedtransempty{\frass{Q} \ftimes \Psi}{P} {\alpha}{P'}}
{\framedtransempty{\Psi}{P \pll Q}{\alpha}{P' \pll Q}}

\inferrule*[left=\textsc{Scope}, right={$b \freshin \alpha,\Psi$}]
    {\framedtransempty{\Psi}{P}{\alpha}{P'}}
    {\framedtransempty{\Psi}{(\nu b)P}{\alpha}{(\nu b)P'}}

\inferrule*[left=\textsc{Open}, right={$\inferrule{}{b \freshin \ve{a},\Psi,M\\\\
b \in \n(N)}$}]
    {\framedtransempty{\Psi}{P}{\bout{\ve{a}}{M}{N}}{P'}}
    {\framedtransempty{\Psi}{(\nu b)P}{\bout{\ve{a} \cup \{b\}}{M}{N}}{P'}}

\inferrule*[left=\textsc{Rep}]
   {\framedtransempty{\Psi}{P \pll !P}{\alpha}{P'}}
   {\framedtransempty{\Psi}{!P}{\alpha}{P'}}

\end{mathpar}
\caption{Operational semantics.
Symmetric versions of \textsc{Com}
and \textsc{Par} are elided. In the rule $\textsc{Com}$ we assume that $\fr{P} =
\framepair{\frnames{P}}{\frass{P}}$ and   $\fr{Q} =
\framepair{\frnames{Q}}{\frass{Q}}$ where $\frnames{P}$ is fresh for all of 
$\Psi, \frnames{Q}, Q, M$ and $P$, and that $\frnames{Q}$ is correspondingly
fresh. In the rule
\textsc{Par} we assume that $\fr{Q} = \framepair{\frnames{Q}}{\frass{Q}}$
where $\frnames{Q}$ is fresh for
$\Psi, P$ and $\alpha$. 
In $\textsc{Open}$ the expression $\tilde{a} \cup \{b\}$ means the sequence
$\tilde{a}$ with $b$ inserted anywhere.
}
\label{table:struct-free-labeled-operational-semantics}
\end{table*}

\begin{definition}[Transitions]

A {\em transition} is of the kind \mbox{$\framedtransempty{\Psi}{P}{\alpha}{P'}$}, meaning that when the environment contains the assertion $\Psi$ the agent $P$
can do an $\alpha$ to become $P'$.  The transitions are defined inductively in 
Table~\ref{table:struct-free-labeled-operational-semantics}. We write $\trans{P}{\alpha}{P'}$ to mean
$\framedtransempty{\emptyframe}{P}{\alpha}{P'}$. In \textsc{In} the substitution
is defined by induction on agents, using  substitution on terms, assertions and
conditions for the base cases and avoiding captures through alpha-conversion in
the standard way.
\end{definition}

Both agents and frames are identified by alpha equivalence. This means that we can choose the bound names fresh in the premise of a rule. In a transition the names in $\bn{\alpha}$ count as binding into both the action object and the derivative, and transitions are identified up to alpha equivalence.
This means that the bound names can be chosen fresh, substituting each
occurrence in both the object and the derivative. This is the reason why
$\bn{\alpha}$ is in the support of the output action: otherwise it could be
alpha-converted in the action alone. Also, for the side conditions in
\textsc{Scope} and \textsc{Open} it is important that $\bn{\alpha} \subseteq
\n(\alpha)$. In rules \textsc{Par} and \textsc{Com}, the freshness conditions
on the involved frames will ensure that if
a name is bound in one agent its representative in a frame is distinct from
names in parallel agents, and also (in \textsc{Par}) that it does not occur on
the transition label. We defer a more precise account of this to
Section~\ref{sect:formalisation}.

The environmental assertions $\Psi \frames \cdots$ in
Table~\ref{table:struct-free-labeled-operational-semantics} express the effect
that  the environment has on the agent: enabling conditions in \textsc{Case},
giving rise to action subjects in \textsc{In} and \textsc{Out} and enabling
interactions in \textsc{Com}. Thus $\Psi$ never changes between hypothesis and
conclusion except for the parallel operator, where an agent is part of the
environment for another agent. In a derivation tree for a transition, the
assertion will therefore increase towards the leafs by application of
\textsc{Par} and \textsc{Com}. If all environmental assertions are erased and
channel equivalence replaced by identity we get the standard laws of the
pi-calculus enriched with data structures.

In comparison to the applied pi-calculus and the concurrent constraint pi calculus one main novelty is the inclusion of environmental assertions in the rules. They are necessary to make our semantics compositional, i.e., 
 the effect of the environment on an agent is wholly captured by the semantics. In contrast, the labelled transitions of 
the applied and the concurrent constraint pi-calculi must rely on an auxiliary structural congruence, containing axioms such as scope extension $(\nu a)(P\pll Q) \equiv (\nu a)P\pll Q$ if $a\freshin Q$. With our semantics such laws are derived rather than postulated. The advantage of our approach is that  proofs of meta-theoretical results such as compositionality are much simpler since there is only the one inductive definition of transitions.

Substitution enters the semantics at one point only: the law {\sc In} which defines the effect of an input. Returning to the substitution laws in Section~\ref{sec:nominal} it is easy to motivate Law~2: it is needed to make sure that alpha equivalent agents have the same transitions. Law~1 has a more involved motivation related to the fact that the objects of transition labels must record all received names, otherwise
we lose the principle of scope extension. To see this, let $\unit \vdash M \sch M$, $b \freshin M,N$, and
\[R = \lin{M}{x}N \sdot x(y) \sdot \nil \;|\; (\nu b)\out{b}c \sdot \nil\]
The only transitions from $R$ are 
\[\trans{R}
{\inn{M}{N\lsubst{L}x}}
{(x(y) \sdot \nil)\lsubst{L}x \;|\;(\nu b)\out{b}c \sdot \nil}\]
for all $L$. Here there is no communication possible between the two components,
even if $L=b$. In contrast, consider
\[T = (\nu b)(\lin{M}{x}N \sdot x(y) \sdot \nil \;|\; \out{b}c \sdot \nil)\]
$T$ is obtained from $R$ through scope extension. Without Law~1 we can have $b \freshin N\lsubst{b}x$ which means that through {\sc Scope} there is a transition
\[\trans{T}
{\inn{M}{N\lsubst{b}x}}
{(\nu b)(b(y) \sdot \nil) \;|\;\out{b}c \sdot \nil)}\]
which can continue with an interaction between the components. $R$ and $T$ therefore do not behave the same. The culprit is the transition from $T$ which corresponds to a scope intrusion, i.e. the reception of a name which is already bound in the receiving agent. To prevent such transitions the law {\sc Scope} has a side condition that the bound name may not occur in the transition label. For this side condition to be effective, Law~1 guarantees that a received name actually appears in the label.

\subsection{Illustrative examples}
\label{sec:illustrativeexamples}
For a simple example of a transition, suppose for an assertion $\Psi$ and condition $\varphi$ that $\Psi \vdash \varphi$.
Assume that
\[\forall \Psi' . \framedtransempty{\Psi'}{Q}{\alpha}{Q'}\]
 i.e., $Q$ has an action $\alpha$ regardless of the environment. Then by the \textsc{Case} rule we get
\[\framedtransempty{\Psi}{\ifthen{\varphi}{Q}}{\alpha}{Q'}\]
 i.e., $\ifthen{\varphi}{Q}$ has the same transition if the environment is $\Psi$. 
 Since $\fr{\pass{\Psi}} = \Psi$ and $\Psi \ftimes \emptyframe = \Psi$, if $\bn{\alpha} \freshin \Psi$  
 we get by \textsc{Par} that
\[\framedtransempty{\emptyframe}{\pass{\Psi}\pll\ifthen{\varphi}{Q}}{\alpha}{\pass{\Psi}\pll Q'}\] 

Data terms may also represent communication channels and here the channel equivalence comes into play. For example, in a polyadic pi-calculus the terms include tuples and projection functions
with the usual equalities, e.g.~$\pi_1(\mathrm{t}_2(a,b))=a$. If these
terms can represent channels then they must represent the {\em same} channel,
consequently we must have
$\Psi \vdash \pi_1(\mathrm{t}_2(a,b)) \sch a$ for all $\Psi$. As an example,  
\[\trans{\out{a}{N}\sdot P \;|\; \inprefixempty{\pi_1(\mathrm{t}_2(a,b))}\,(y)\sdot Q}{\tau}{P\pll Q\lsubst{N}y}\]
 Agents such as $ \inprefixempty{\pi_1(\mathrm{t}_2(a,b))}\,(y)\sdot Q$ can arise naturally if tuples of channels are transmitted as objects. For example, an agent that receives a pair of channels along $c$ and then inputs along the first of them is written $\inprefixempty{c}(x) \sdot \inprefixempty{\pi_1(x)}(y) \sdot Q$. When put in parallel with an agent that sends $\mathrm{t}_2(a,b)$ along $c$ it will have a transition leading to the agent where $x$ is substituted by $\mathrm{t}_2(a,b)$, i.e. 
 $\inprefixempty{\pi_1(\mathrm{t}_2(a,b))}\,(y)\sdot Q$.

The semantics makes no particular provision for an equality of terms in object position. Thus, the agents $\out{c}{a}\sdot P$ and $\out{c}{\pi_1(\mathrm{t}_2(a,b))}\sdot P$
 have different transitions, and correspond to sending out the unevaluated ``texts'' $a$ and $\pi_1(\mathrm{t}_2(a,b))$ respectively. To represent agents which send evaluated ``values'' we can do as in the applied pi-calculus where assertions declare equivalence of terms and agents send freshly generated aliases, e.g. 
\[(\nu z) (\out{c}{z}\sdot P \;|\; \pass{z = \pi_1(\mathrm{t}_2(a,b))})\]
 This agent has the same transition as  $(\nu z) (\out{c}{z}\sdot P \;|\; \pass{z = a})$. Any agent receiving the $z$ will not be able to distinguish if $z$ is $a$ or $\pi_1(\mathrm{t}_2(a,b))$ since these terms are equated by all assertions. Also, if $a$ and $b$ are scoped as in
 \[(\nu a,b,z) (\out{c}{z}\sdot P \;|\; \pass{z = \pi_1(\mathrm{t}_2(a,b))})\]
 then their scopes will {\em not} open as a consequence of the output.  
 In the applied pi-calculus this is the only form of communication and it is not possible to directly transmit data structures containing channel names, like the name tuples of the polyadic pi-calculus  above. In psi-calculi these communication possibilities can coexist.

The main technical issue in the semantics is the treatment of scoping, as illustrated by the following example where the terms are just names.
The intuition is that there is a communication channel available to all agents, and agents can declare any name to represent it through an assertion. The assertions are thus sets of names, and any name occurring in the assertion can be used as the subject of an action. Any two names in the assertion are deemed channel equivalent. Formally,

\begin{mathpar}
\begin{array}{rcl}
\terms & \defn & \nameset\\
\conditions & \defn & \{a\sch b : a,b\in \terms\}\\
\assertions & \defn & \mathcal{P}_{\rm fin}(\nameset)\\
\ftimes & \defn & \cup \\
\emptyframe & \defn & \emptyset \\
 \vdash & \defn & \{(\Psi, a \sch b) : a,b \in \Psi\} 
 \end{array}
\end{mathpar}
Omitting the action and prefix objects we get
\[\framedtransempty{\{a,b\}}{\overline{a}\sdot \nil}{\overline{a}}{\nil}\]
and also
 \[\framedtransempty{\{a,b\}}{\overline{a}\sdot \nil}{\overline{b}}{\nil}\]
 By the
\textsc{Par} rule 
 we have
 \[\framedtransempty{\emptyset}{\overline{a}\sdot \nil\pll\pass{\{a,b\}}}{\overline{a}}{
\nil\pll\pass{\{a,b\}}}\] and
 \[\framedtransempty{\emptyset}{\overline{a}\sdot \nil\pll\pass{\{a,b\}}}{\overline{b}}{
\nil\pll\pass{\{a,b\}}}\]
 Applying a restriction we get
 \[\framedtransempty{\emptyset}{(\nu
a)(\overline{a}\sdot \nil\pll\pass{\{a,b\}})}{\overline{b}}(\nu
a)({\nil\pll\pass{\{a,b\}}})\]
 but no corresponding action with subject $a$ because of the side condition on
$\textsc{Scope}$. 
 Thus, a communication through $\textsc{Com}$ can be inferred from 
 \[(\nu
a)(\overline{a}\sdot \nil\pll\pass{\{a,b\}})\parop b\sdot \nil\]
 but {\em not} from \[(\nu
a)(\overline{a}\sdot \nil\pll\pass{\{a,b\}}) \parop a\sdot \nil\]

This instance of a psi-calculus also illustrates two features of the semantics: firstly that  channel equivalence is used in all three rules $\textsc{In}$, $\textsc{Out}$ and $\textsc{Com}$, and secondly that assertions rather than frames represent the environment. Both issues are related to the law of 
scope extension.
Elaborating the example above and noting that $\{a\} \cup
\{b\} \vdash a \sch b$, we get that
\[(\nu a,b)(\pass{\{a\}}\parop\pass{\{b\}}\parop\overline{a}\sdot \nil\parop
b\sdot \nil)\]
 has an internal communication. By scope extension
this agent should have the same transitions as $P\parop Q$ where
\[P=(\nu a)(\pass{\{a\}}\parop\overline{a}\sdot \nil)\qquad Q=(\nu b)(\pass{\{b\}}
\parop b\sdot \nil)\]
Here $\fr{P} =  \framepair{a}{\{a\}}$ and 
$\fr{Q} =  \framepair{b}{\{b\}}$ are alpha equivalent. Since they will be composed below we choose different representatives for the bound names.
A communication from $P\parop Q$ is inferred by \textsc{Com} and the premises
\[\begin{array}{l}
1. \;\framedtransempty{\{b\}}{P}{\overline{b}}{(\nu
a)(\pass{\{a\}}\parop\nil)} \\
\qquad
\mbox{(derived using $\{a\}\rel{\ftimes}\{b\}=\{a, b\}\vdash a\sch b$ in {\sc Out)}} \\
2. \; \framedtransempty{\{a\}}{Q}{a}{(\nu
b)(\pass{\{b\}}\parop\nil)}\\
\qquad
\mbox{(derived using $\{b\}\rel{\ftimes}\{a\}=\{a, b\}\vdash a\sch b$ in {\sc In)}} \\
3. \;\{a\}\rel{\ftimes}\{b\}= \{a, b\}\vdash a\sch b
\end{array}\]
Note how the action subjects are derived by the assertions in both cases to not clash with the binders, and that channel equivalence is necessary in all three rules. \label{sec:channel-equivalence-everywhere}

The same example demonstrates why transitions in
Table~\ref{table:struct-free-labeled-operational-semantics} are defined with
assertions and not frames, for whereas $\{a, b\}\vdash a\sch b$ the corresponding result cannot be obtained from the frames of the agents. We have that
$\fr{Q}\rel{\ftimes}{\{a\}}= \framepair{b}{\{a,b\}} \nvdash a \sch b$, so that
frame is not useful for deriving a transition from $P$.
Our earlier attempt~\cite{johansson.parrow.ea:extended-pi} erroneously used frames rather than assertions, and this means that scope extension does not hold unless a further condition is imposed on the entailment relation to eliminate this kind of example.

We close this section by demonstrating why the requisites in Definition~\ref{def:entailmentrelation} are necessary: omitting any of them would lead to a calculus that does not satisfy fundamental properties of the parallel operator.
Compositionality and the abelian monoid laws in Definition~\ref{def:entailmentrelation} are straightforward in this respect since without them the corresponding properties of parallel composition on agents do not hold. For example, we will want parallel composition to be commutative in that the agent $P \pll Q$ behaves the same as  $Q\pll P$ in all respects. At the very least this implies that their frames entail the same conditions (it may also imply other things not important for this argument), which means that $\ftimes$ must be commutative for assertion equivalence. In a similar way the other requisites on $\ftimes$ are necessary for parallel operator to be compositional, associative, and have $\nil$ as identity.

To demonstrate that channel equivalence must be symmetric, consider any psi-calculus where $\Psi_1$ and $\Psi_2$ are such that $\Psi_1 \ftimes \Psi_2
\vdash a \sch b$ and $\Psi_1 \ftimes \Psi_2 \vdash b \sch b$. We shall argue that also $\Psi_1 \ftimes \Psi_2 \vdash b \sch a$ must hold, otherwise scope extension does not hold. 
Consider the agent
\[(\nu a,b)(\pass{\Psi_1}\parop\pass{\Psi_2}\parop\overline{a}\sdot \nil\parop
b\sdot \nil)\]
which has an internal communication $\tau$ using $b$ as subjects in the premises of the {\sc Com} rule. 
If $b \freshin \Psi_1$ and $a \freshin \Psi_2$, by scope extension the agent should behave as
\[(\nu a)(\pass{\Psi_1}\parop\overline{a}\sdot \nil) \;\parop\; (\nu b)(\pass{\Psi_2}\parop
b\sdot \nil)\]
and therefore this agent must also have a $\tau$ action. The left hand component cannot do an $\overline{a}$ action, but in the environment of $\Psi_2$ it can do a $\overline{b}$ action.
Similarly, the right hand component cannot do a $b$ action. The only possibility is for it to do an $a$ action, as in
\[\framedtransempty{\Psi_1}{(\nu
b)(\pass{\Psi_2} \parop b \sdot \nil)}{a}{\cdots}\]
and this requires $\Psi_1 \ftimes \Psi_2 \vdash b \sch a$.

Finally, we motivate the requisite that $\sch$
must be transitive. Let $\emptyframe$ entail $a \sch a$ for all names $a$, and
let $\Psi$ be an assertion with support $\{a,b,c\}$ that additionally entails
the two conditions $a \sch b$ and $b \sch c$, but not $a \sch c$, and thus does
not satisfy transitivity of channel equivalence. If $\Psi$ entails no other
conditions then $(\nu b)\Psi \sequivalent \emptyframe$, and we expect $(\nu b)
\pass\Psi$ to be interchangeable with $\pass\emptyframe$ in all contexts.
Consider the agent
\[\overline{a}\sdot \nil \parop c \sdot \nil \parop (\nu b)\pass{\Psi}\]
By scope extension it should behave precisely as
\[(\nu b)(\overline{a}\sdot \nil \parop c \sdot \nil \parop \pass\Psi)\]
This agent has a $\tau$-transition since  $\Psi$ enables an interaction between the components $\overline{a}\sdot \nil$ and $c \sdot \nil$.
But the agent
\[\overline{a}\sdot \nil \parop c \sdot \nil \parop \pass{\emptyframe}\]
has no such transition. The conclusion is that $(\nu b) \Psi$ must entail that the components can communicate, ie. that $a \sch c$, in other words $\Psi \vdash a \sch c$.

\section{Expressiveness and related calculi}

\label{sec:expressiveness}
In this section we explore the expressiveness of psi-calculi, mainly in comparison to other process calculi.

\subsection{The pi-calculus}
In Section \ref{sec:agents} we saw the instance \piinstance{} which
corresponds to the pi-calculus. We will now make the relationship formal. The
pi-calculus under consideration is the standard pi-calculus with replication instead of
recursion, without mismatch, and without a rule for structural congruence in
the semantics. %
The encoding of a pi-calculus agent $P$ into \piinstance{}, $\encodepi{P}$, is
defined as:
\[
\begin{array}{rll}
\encodepi{\nil} &=& \nil\\
\encodepi{\outprefix{a}{b} \sdot P} &=& \outprefix{a}{b} \sdot \encodepi{P}\\
\encodepi{\inprefixpi{a}{x}\sdot P} &=& \inprefix{a}{x}{x} \sdot \encodepi{P}\\
\encodepi{P \pll Q} &=& \encodepi{P} \pll \encodepi{Q}\\
\encodepi{!P} &=& !\encodepi{P}\\
\encodepi{(\nu a)P} &=& (\nu a)\encodepi{P}\\
\encodepi{[a=b]P} &=& \case{\ci{a=b}{\encodepi{P}}}\\
\encodepi{P + Q} &=& (\nu a)(\case{\ci{a=a}{\encodepi{P}} \casesep
\ci{a=a}{\encodepi{Q}}}) \text{ where } a \freshin P,Q
\end{array}\]
To prove that $P$ and $\encodepi{P}$ 
have the same transitions the following two
lemmas about substitutions and support are needed. We use the standard
definition of substitution in the pi-calculus, replacing free names for new
ones, $\alpha$-converting as necessary to avoid capture.
\begin{lemma} If $P$ is a \pic{} agent, then
 $\encodepi{P}\lsubst{b}{x} = \encodepi{P\lsubst{b}{x}}$.
\end{lemma}
\begin{proof}
Straightforward induction over the structure of $P$.
\end{proof}
\begin{lemma}
\label{lemma:encoding-has-same-support}
If $P$ is a \pic{} agent, then
 $\names{P} = \names{\encodepi{P}}$. 
\end{lemma}
\begin{proof}
Straightforward induction over the structure of $P$.
\end{proof}

Let $\alpha$ be a \pic{} action. We define the encoding of $\alpha$
into \psicalculi{} actions as:
\[\begin{array}{rcl}
 \encodepi{\outlabel{a}{b}} &=& \outlabel{a}{b}\\
 \encodepi{\bout{b}{a}{b}} &=& \bout{b}{a}{b}\\
 \encodepi{\inlabelpi{a}{b}} &=& \inlabel{a}{b}\\
 \encodepi{\tau} &=& \tau
\end{array}
\]

We denote a pi-calculus transition as $\transpi{P}{\alpha}{P'}$.
We then have the following relation between the \pic{} agent $P$ and the
\piinstance{} agent $\encodepi{P}$:
\begin{lemma}[Transitions in \piinstance{} and the \pic{} correspond]
\label{lemma:correspondence-pi-encoding}
If $P$ is a \pic{} agent, then
\[
 \text{ if }\transpi{P}{\alpha}{P'} \text{ then }
\trans{\encodepi{P}}{\encodepi{\alpha}}{\encodepi{P'}}
\] and
\[
 \text{ if }\trans{\encodepi{P}}{\alpha'}{P''} \text{ then
}\transpi{P}{\alpha}{P'}\text{ where }\encodepi{\alpha} = \alpha' \text{ and }
\encodepi{P'}=P''.
\]
\end{lemma}
\begin{proof}
 The proof is by induction over the length of the derivation of
$\transpi{P}{\alpha}{P'}$ and $\trans{\encodepi{P}}{\alpha}{P''}$,
respectively. As an illustration, one induction case is shown:
the case when the
\pic{} transition is derived with \textsc{$\pi$-Close}:
\[
 \inferrule*[left=\textsc{$\pi$-Close}, right=$b \notin \fn{Q}$]
   {\transpi{P}{\bout{b}{a}{b}}{P'} \\
     \transpi{Q}{\inlabelpi{a}{b}}{Q'}}
   {\transpi{P \pll Q}{\tau}{(\nu b)(P' \pll Q')}}
\]
By induction it follows that
$\trans{\encodepi{P}}{\bout{b}{a}{b}}{\encodepi{P'}}$ and that
$\trans{\encodepi{Q}}{\inlabel{a}{b}}{\encodepi{Q'}}$. Since there is only one
assertion in \piinstance{}, the frames of $\encodepi{P}$ and $\encodepi{Q}$
will be equivalent to $1$.
We choose the frames so that their bound names are %
sufficiently fresh according to rule \textsc{Com}. It trivially holds that
$1 \vdash a=a$, and by definition in \piinstance{} we have that $1 \ftimes 1 =
1$, so also $1 \ftimes 1 \ftimes 1 \vdash a=a$. Since $b \notin \fn{Q}$
(i.e. $b \freshin Q$) it follows from
Lemma~\ref{lemma:encoding-has-same-support} that $b \freshin \encodepi{Q}$. We
now derive the following:
\[
 \inferrule*[left=\textsc{Com}, right=$b \freshin \encodepi{Q}$]
   {1 \ftimes 1 \frames \trans{\encodepi{P}}{\bout{b}{a}{b}}{\encodepi{P'}} \\
    1 \ftimes 1 \frames \trans{\encodepi{Q}}{\inlabel{a}{b}}{\encodepi{Q'}} \\
    1 \ftimes 1 \ftimes 1 \vdash a=a}
   {1 \frames \trans{\encodepi{P} \pll \encodepi{Q}}{\tau}
                    {(\nu b)(\encodepi{P'} \pll \encodepi{Q'})}
}
\]
By definition we have that $\encodepi{P \pll Q} = \encodepi{P} \pll
\encodepi{Q}$, and that $\encodepi{(\nu b)(P' \pll Q')} = (\nu b)(\encodepi{P'}
\pll \encodepi{Q'})$, and that $\encodepi{\bout{b}{a}{b}} =
\bout{b}{a}{b}$, so in other words we have that $\trans{\encodepi{P \pll
Q}}{\encodepi{\tau}}{\encodepi{(\nu b)(P' \pll Q')}}$.
\end{proof}

In Section~\ref{sec:bisimilarity} we shall see that strong bisimulation in the
pi-calculus and in \piinstance{} coincide.

\subsection{Calculi for cryptography}
\label{sect:api}

Psi-calculi can express  a variety of 
cryptographic operations on data.
The main idea was illustrated in Section~\ref{sec:first-crypto-example}, using 
  assertions to define relations between ciphertext and plaintext.  Here we make the description more precise.
Let the assertion
``$C=\sym{enc}{M,k}$" mean that encrypting the message $M$ with the key $k$ results in the ciphertext $C$, and let
``$M = \sym{dec}{C,k}$" mean that decrypting $C$ with key $k$ yields $M$. Entailment contains equations relating encryption and decryption such as $\forall M,k .\,\sym{dec}{\sym{enc}{M,k},k}=M$. The point is that a secure key can be represented by a bound name: it is unguessable outside its scope. An example agent
$\op{a}C\sdot\res{k}(\pass{C=\sym{enc}{M,k}}\parop P)$ outputs a term
$C$ and asserts that it is the encryption of $M$ using the bound $k$ as
key, without opening the scope of $k$.
Therefore an agent receiving $C$ can resolve the condition
$\sym{dec}{C,k}=M$ only  after
receiving this $k$ in a communication. Technically this is because of the
freshness conditions in the \textsc{Par} rule in
Table~\ref{table:struct-free-labeled-operational-semantics} where $\frnames{Q}$
is assumed fresh for $P$: this means that to apply the rule, $P$ cannot use any
name bound in the frame of $Q$.

This closely resembles the situation in the applied pi-calculus~\cite{abadi.fournet:mobile-values}. By contrast, in the spi-calculus~\cite{abadi.gordon:calculus-cryptographic} encrypted messages such as
$\sym{enc}{M,k}$ are transmitted directly.
Consider an example spi-calculus process
\begin{equation}
P=\res{k,m}\opa{a}{\sym{enc}{m,k}}\sdot P' \qquad\text{where }P'=b(x)\sdot\ifthen{x=m}{\op{c}}\label{eqn:p}
\end{equation}
Here $P$ sends a fresh name $m$ encrypted with a
fresh key $k$ to the environment, and then receives a value
$x$. Assuming perfect encryption, the environment cannot know $m$ or
$k$, so $P'$ cannot receive $m$ along $b$, and the output on $c$ will never
be possible. However, in the spi-calculus the transition
\mbox{$P\gt{\res{k,m}\opa{a}{\sym{enc}{m,k}}}P'$}  opens the scopes of $k$ and $m$,
so here scoping does not correspond to restriction of knowledge.
A reasonable
equivalence must explicitly keep track of which names are known,
leading to several complex bisimulation definitions
(see~\cite{borgstrom.nestmann:bisimulations-spi} for an overview).

The applied pi-calculus is 
data terms and an equational theory $\vdash_\Sigma$ over $\Sigma$,
and, more importantly, introduces \emph{active substitutions}
$\subst{M}{x}$ of data terms for variables. These can be introduced by
the inferred structural rule $\res{x}(\subst{M}{x}\parop A)\equiv
A[x:=M]$.
There are names $a,b,c$ distinct from variables $x,y,z$ where only
variables can be substituted, and a simple type system to
distinguish names and variables of channel type from other terms of base
type. Only names of channel type can be used as communication
channels.
Structured data terms cannot be sent directly, instead an \emph{alias}
variable such as $x$ must be used, and the term itself does not occur
on the transition label. We have
$P\equiv Q$ for $P$ above in (\ref{eqn:p}), where
\begin{equation}
Q=\res{x,k,m}(\subst{\sym{enc}{m,k}}{x}\parop\op{a}x\sdot P')\label{eqn:q}
\end{equation}
Here $Q\gt{\out{a}{\res{x}x}}\res{k,m}(\subst{\sym{enc}{m,k}}{x}\parop P')$
and only the alias of the encryption (its ``value'') appears on the label; the
scope of $k$ and $m$ is not opened and 
in this sense they are still confidential to the environment.
However,
the labelled semantics does not allow sending structured data terms where the
scope \emph{should} be opened, such as a tuple of names in the
polyadic pi-calculus. 
\label{sect:apianomaly}
The labelled semantics for applied pi
turns out to be non-compositional. 
Consider the closed (extended) applied pi-calculus agents
\begin{equation}
A=(\nu a)(\subst{a}x\parop x.b.\nil) \qquad B=(\nu a)(\subst{a}x\parop\nil)\label{eq:api-comp}
\end{equation}
 where we omit the objects of the prefixes. They have the same frame and no transitions, and are thus semantically equivalent. But a context can contain $x$ and can therefore use the active substitution to communicate with $A$. Formally, let $R = \overline{x}.\nil$ and $\Downarrow b$ the usual weak observation or barb.  We have by scope extension that 
$A\pll R \equiv (\nu a)(\subst{a}x\parop x\sdot
b\sdot\nil\parop\overline{x}\sdot\nil) \Downarrow b$, but it is not the case
that $B\pll R \Downarrow b$.
Therefore, {\em no} observational equivalence that is preserved by
all contexts and satisfies scope extension can be captured by the
labelled semantics. In this,
Theorem 1 of \cite{abadi.fournet:mobile-values} is incorrect; the
labelled and observational equivalences do in fact not coincide, nor
is labelled equivalence a congruence. This is relevant for
other papers that use or
develop the labelled semantics,
e.g.~\cite{goubault-larrecq.palamidessi.ea:probabilistic-applied,KremerRyan2005,DKR-fsttcs07,cortier.rusinowitch.ea:relating-two,godskesen:observables-wireless}.

Possible fixes are to disallow aliases for channel names, to be satisfied with compositionality for closed contexts, or to allow variables in action subjects. The consequences are difficult to assess, and our proposed solution is to instead define a psi-calculus.

A complication when defining a psi-calculus to correspond to the applied pi-calculus is that bisimulation there is only defined on closed agents, and removing this restriction yields a non-compositional theory. The source of this non-compositionality is the requirement that active substitutions must be acyclic. Assume that the equational system includes the identity $\mathsf{f}(y) = \mathsf{f}(z)$. We then get that $\subst{\mathsf{f}(y)}{x}$ is bisimilar to $\subst{\mathsf{f}(z)}{x}$, but only one becomes circular when composed with $\subst{x}{y}$. In psi-calculi, no notion of closedness exists, and compositionality is required. For these reasons we cannot exactly capture the applied pi-calculus. %

We define the instance \apiinstance{} as follows (this presentation
corrects a mistake in~\cite{DBLP:conf/lics/BengtsonJPV09}). Since our
names and terms are untyped we add constructs for channels,
$\chanCon{M}$, for variables, $\varCon{x}$, and for names which are
neither channels nor variables, $\nonceCon{k}$.  We extend
$\vdash_\Sigma$ so that $\vdash_\Sigma \chanCon{M}=\chanCon{M}$ for
all $M\in\terms$, $\vdash_\Sigma \nonceCon{a} = \nonceCon{a}$ for all
$a \in \nameset$, and $\vdash_\Sigma \varCon{x} = \varCon{x}$ for all
$x \in \nameset$.  Furthermore we define
$\eqs{\subst{M_1}{N_1},\dots,\subst{M_n}{N_n}}$ to be the set of
equations $\{M_1 = N_1, \dots, M_n = N_n\}$. Substitution on terms is
defined in the expected way except for terms of kind $\varCon{x}$ and
$\nonceCon{a}$. For terms of these kinds we have that
$\varCon{x}\lsubst{M}{x} = M$ and $\nonceCon{a}\lsubst{M}{a} = M$.
A term $M$ is \emph{ground} if it has no subterms of kind $\varCon{x}$. We write $\vdash_{\Sigma \cup \Sigma'}$ for the equational theory $\vdash_\Sigma$ extended with the equations from $\vdash_{\Sigma'}$.
\begin{mathpar}
\begin{array}{rcl}
\terms & \defn & \nameset \cup \{\nonceCon{k} : k \in \nameset\} \cup
\{\varCon{x} : x \in \nameset\} \cup \{\chanCon{M} : M \in \terms \} \cup\\
&&\{\mathsf{f}(M_1,\dots, M_n) : \mathsf{f} \in \Sigma \wedge M_i \in \terms\}\\
\conditions & \defn & \{M=N, \neg(M=N),M\sch N: M,N\in\terms\}\\
\assertions & \defn & \mathcal{P}_{\rm fin}(\{\subst{M}{N}: M,N\in\terms\})\} \\
\ftimes & \defn & \cup
\\
\emptyframe & \defn &\emptyset\\
\multicolumn{3}{l}{\Psi\vdash M=N\text{ if }\vdash_{\Sigma\cup\eqs{\Psi}} M =
N}\\
\multicolumn{3}{l}{\Psi\vdash \neg(M=N)\text{ if }\begin{array}[t]{@{}l@{}}
	\text{there exists ground }M',N'\text{ such that }\\
	\vdash_{\Sigma \cup \eqs{\Psi}} M = M' \text{, }\\
	\vdash_{\Sigma \cup \eqs{\Psi}} N=N'\text{, and }\\
	\neg(\Psi\vdash M'=N') \\
						 \end{array}}\\
\multicolumn{3}{l}{\Psi\vdash M\sch N\text{ if }\Psi\vdash M=N\wedge\exists c : \Psi\vdash M=\chanCon{c}}\\
\end{array}
\end{mathpar}

Assertions are finite sets of active substitutions of the more general form $\subst{M}{N}$, $\ftimes$ is union, and entailment deduces equality under the equational theory with equations added to represent the active substitutions.
The conditions are as for the applied pi-calculus except for 
$\neg(M=N)$ which is needed to represent the
$\ifthel{M=N}{P}{Q}$ construct of applied pi as
$\caseonly{\ci{M=N}{P}\casesep\ci{\neg{(M=N)}}{Q}}$ in \apiinstance{}. As in applied pi, the terms compared
for inequality need to be ground.
Channel equivalence $M\sch N$ requires that there is a channel name equal to
both $M$ and $N$.

To see that this is a proper instance we must check that the substitution function is equivariant and respects the freshness and $\alpha$-equivalence properties, as described in Section~\ref{sec:nominal}. Furthermore it must satisfy the requirements from Definition~\ref{def:entailmentrelation}. That the substitution function has the required properties is shown in Section~\ref{sec:nominal}, and the special cases for $\varCon{x}$ and $\nonceCon{a}$ pose no additional problem.
Channel symmetry and transitivity hold since the underlying equational theory is symmetric and transitive. Identity, associativity, and commutativity hold since union has these properties. Compositionality holds assuming that the equational system is compositional, i.e if $\forall M,N :\;\vdash_{\Sigma_1} M=N \Leftrightarrow\; \vdash_{\Sigma_2} M=N$ implies $\forall M,N :\;\vdash_{\Sigma_1 \cup \Sigma'} M=N \Leftrightarrow\; \vdash_{\Sigma_2 \cup \Sigma'} M=N$.

The encoding $\encodeapi{A}$ of an applied pi agent $A$ into \apiinstance{} is homomorphic with the following exceptions:
\[
\begin{array}{rcl}
\encodeapi{a} &=& \chanCon{a}\text{ if the name }a\text{ is of channel type and not a binding occurrence}\\
\encodeapi{x} &=& \varCon{x}\text{ if the variable }x\text{ is not a binding occurrence}\\
\encodeapi{k} &=& \nonceCon{k}\text{ if the name }k\text{ is not of channel type or a binding occurrence}\\
\encodeapi{\subst{M}{x}} &=& \subst{\encodeapi{M}}{\varCon{x}}
\end{array}
\]
Note that in translations of applied pi-calculus agents and their derivatives, the only form of active substitutions will be on the form $\subst{M}{\varCon{x}}$.
Also the only substitutions will be of variables.
We allow for the general form of active substitutions $\subst{M}{N}$ and substitution of channels and nonces simply to make the substitution function total as required. We adhere to the applied pi convention that channel names are ranged over by $a,b,c,\dots$, nonces are ranged over by $k,l,m,\dots$, and variables are ranged over by $x,y,z,\dots$. For readability, in the following we omit the constructs $\chanCon{a}$, $\nonceCon{k}$, and $\varCon{x}$, and just write $a$, $k$, and $x$, also in \apiinstance{}-agents. 

\apiinstance{} differs
from the applied pi-calculus in some ways.
Requirements on the active substitutions in applied pi are that they can only
contain one active substitution per variable,  and that the active
substitutions are non-circular. Furthermore they do not occur under prefixes, conditionals, or replication.
The instance \apiinstance{} does not have these limitations,
but the most important difference is that in \apiinstance{} (and in psi-calculi in general), $\op{a}M\sdot P\gt{\op{a}M}P$ 
corresponds to sending the cleartext of $M$ directly. This is not possible in the applied pi-calculus. 
In order to transmit $M$ in the applied pi-calculus the structural rule 
$\res{x}(\subst{M}{x}\parop A)\equiv A[x:=M]$
must be used and
an alias $x$ for
$M$ be sent. To send an alias in this way in \apiinstance{} it must be introduced explicitly, as in $\pass{\subst{M}{x}}\parop\op{a}{x}\sdot P$, and this agent is {\em not} the same as  $\op{a}M\sdot P$.
Therefore, although the agents \[P=\res{k,m}\opa{a}{\sym{enc}{m,k}}\sdot P'\] and \[Q=\res{x,k,m}(\subst{\sym{enc}{m,k}}{x}\parop\op{a}x\sdot P')\] (from equations (\ref{eqn:p}) and (\ref{eqn:q})) are the same in the applied pi-calculus,
the \apiinstance{} counterparts of the
agents are different. In \apiinstance{}, $P$ in (\ref{eqn:p}) represents an agent that     
emits the cleartext ``$\sym{enc}{m,k}$''. Any agent that  receives this will immediately learn both $m$
and $k$, and any scope of $k$ will be opened in the process. This kind of agent
can only indirectly be represented in the applied pi-calculus, by sending the
restricted names separately one at a time.
In contrast, the \apiinstance{} counterpart of (\ref{eqn:q}) is 
$Q=\res{x,k,m}(\pass{\subst{\sym{enc}{m,k}}{x}}\parop\op{a}x\sdot P')$
and defines $Q$ to emit an alias for $\sym{enc}{m,k}$.
As in the applied pi-calculus since $k$ is scoped a recipient will not learn $m$. 
If the same recipient later receives $k$, an alias $u$ for the message $m$ can be constructed as $\pass{\subst{\sym{dec}{x,k}}{u}}$.
Similarly, the agents $R_1$ and $R_2$ below are equivalent in applied pi, but the corresponding agents in \apiinstance{} are different.
\[
\begin{array}{rcl}
R_1&=&\res{x,k,m}(\subst{\sym{enc}{m,k}}{x}\parop\op{a}x\sdot \op{a}x\sdot P')\\
R_2&=&\res{x,k,m}(\subst{\sym{enc}{m,k}}{x}\parop \res{y}(\subst{x}{y}\parop \op{a}x\sdot \op{a}y\sdot P'))\\
\end{array}
\]
In the applied pi-calculus, a new alias for a term can always be introduced ``on-the-fly'', and it is impossible to tell $R_1$ and $R_2$ apart -- they are structurally equivalent. 
The psi-calculus approach gives the possibility to discern the two agents, similarly to how the same ciphertext bitstring sent twice can be identified even if the plaintext cannot be recovered. 
To avoid this, a new alias needs to be explicitly introduced for each transmission, mimicking a probabilistic 
crypto where different ciphertext bitstrings correspond to the same plaintext and key.

Thus in psi-calculi,
communication objects can range from literal data terms to
indirect references, giving the user of the
calculus the possibility to choose the appropriate form.

Another difference between the calculi is illustrated by the agent $A$ of the
compositionality counterexample (\ref{eq:api-comp}): Its counterpart $P_A$ in \apiinstance{} is
$(\nu a)(\pass{\subst{a}x}\parop x.b.\nil)\gt{x}\res{a}(\pass{\subst{a}x}\parop b.\nil)$ and
is not equivalent to $(\nu a)(\pass{\subst{a}x}\parop\nil)$; indeed also $P_A\parop\overline{x}.\nil\gtt\gt{b}$ in our labelled
semantics. 

In Section~\ref{sec:crypto-example} we present a simpler psi-calculus for expressing cryptographic examples.

\pagebreak[4]
\subsection{Fusion and concurrent constraints}

\subsubsection{Fusion}\label{sec:fusion}

\label{sec:ccpi}
The concept of \emph{fusion} means that communication can result in
pairs of names being ``fused together'' in the sense that they can
thereafter be considered the same.
Fusion was independently developed by
Fu~\cite{fu:proof-theoretical-approach} (the $\chi$-calculus), Parrow and
Victor~\cite{parrow.victor:fusion-calculus} (the fusion calculus), and by Wischik and
Gardner~\cite{gardner.wischik:explicit-fusions-mfcs,gardner.wischik:explicit-fusions} (the pi-F calculus). 
The fusion primitive was also encoded in the asynchronous pi-calculus by Merro~\cite{merro:expressiveness-chi}, using equators.
In psi-calculi, fusion can be formulated in a way reminiscent of the
equator encoding: the assertions are equivalence statements between
names (cf. explicit fusions or equators). A simple psi-calculus with fusion, call it \textbf{Fi}, would be the following:
\begin{mathpar}
\begin{array}{rcl}
\terms & \defn & \nameset\\
\conditions & \defn & \{a=b : a,b\in \terms\} \\
\assertions & \defn & \{\,\{a_1=b_1, \ldots, a_n=b_n\} : a_i \in \nameset, b_i
\in \nameset\}\\
\sch &\defn & =\\
\ftimes & \defn & \cup \\
\emptyframe & \defn & \emptyset \\
\multicolumn{3}{l}{\Psi\vdash a=b\text{ if }(a,b)\in\textsc{eq}(\Psi)}\\
\end{array}
\end{mathpar}
where $\textsc{eq}(\Psi)$ is the equivalence closure of $\Psi$ (i.e. transitive, symmetric and reflexive closure).
Thus terms are names,
assertions are name fusions, and the entailment relation deduces equality
between names based on fusion assertions treated as equivalence
relations.
We can verify that this is indeed a valid psi-calculus: the substitution properties are proved in Section~\ref{sec:nominal}, and we just need to investigate the requisites of Definition~\ref{def:entailmentrelation}. 
Transitivity and reflexivity of the channel equivalence follows from the same properties of $=$; commutativity, associativity and identity follow from the same properties of $\cup$.
For compositionality, let $\Psi_1$ and $\Psi_2$ be two equivalent assertions. This means $\textsc{eq}(\Psi_1)=\textsc{eq}(\Psi_2)$; we must show that for any $\Psi_3$ we have $\textsc{eq}(\Psi_1\cup\Psi_3)=\textsc{eq}(\Psi_2\cup\Psi_3)$. 
Using the fact that $\textsc{eq}(A\cup B)=\textsc{eq}(\textsc{eq}(A)\cup B)$, we
have $\textsc{eq}(\Psi_1\cup\Psi_3)=\textsc{eq}(\textsc{eq}
(\Psi_1)\cup\Psi_3)=\textsc{eq}(\textsc{eq}(\Psi_2)\cup\Psi_3)=\textsc{eq}
(\Psi_2\cup\Psi_3)$.

In the $\chi$-calculus, fusion calculus, and pi-F calculus, input and output prefixes are completely symmetric and in particular the input is not binding. An example transition in the pi-F calculus (using the syntax of~\cite{wischik:thesis}) is
\(
a\vec{b}\sdot P\parop \op{a}\vec{d}\sdot Q\gtt \fusion{\vec{b}}{\vec{d}}\parop P\parop Q
\)
where $\fusion{\vec{b}}{\vec{d}}$ (for $\vec{b}$ and $\vec{d}$ of equal
length) is a \emph{fusion} which allows us to treat
each $b_i\in\vec{b}$ as equivalent to $d_i\in\vec{d}$. 
Inputs in \textbf{Fi} can still be binding, and
  we can represent the non-binding pi-F input
$ a\vec{b}\sdot P$ as
$a(\vec{c})\sdot(\pass{\{\vec{b}=\vec{c}\}}\parop P)$
 where
$\vec{c}\freshin a\vec{b}\sdot P$. 
For example, the pi-F communications 
\[
a b\sdot \op{c}c\sdot P\parop \op{a}{c}\sdot b d\sdot Q\gtt \fusion{b}{c}\parop \op{c} c\sdot P\parop b d\sdot Q\gtt \fusion{b}{c}\parop \fusion{c}{d}\parop P\parop Q
\]
are expressed as:
\begin{mathpar}
\begin{array}{cl}
\multicolumn{2}{l}{a({e})\sdot\big(\pass{\{{b}={e}\}}\parop \op{c}c\sdot P\big)\;\parop\; \op{a}{c}\sdot b(x)\sdot\big(\pass{\{x=d\}}\parop Q\big)}\\
\gtt & \big(\pass{\{{b}={e}\}}\parop \op{c}c\sdot P\big)\lsubst{{c}}{{e}}\;\parop\; b(x)\sdot\big(\pass{\{x=d\}}\parop Q\big)\\
= & \pass{\{{b}={c}\}}\parop \op{c}c\sdot P \;\parop\; b(x)\sdot\big(\pass{\{x=d\}}\parop Q\big)\\
\gtt & \pass{\{{b}={c}\}}\parop P \;\parop\; \big(\pass{\{x=d\}}\parop Q\big)\lsubst{c}{x}\\
= & \pass{\{{b}={c}\}}\parop P \;\parop\; \pass{\{c=d\}}\parop Q
\end{array}
\end{mathpar}

Below, we establish an operational correspondence between the pi-F calculus and \textbf{Fi}.  Our presentation does not include the full details of the pi-F semantics, instead we refer to~\cite{gardner.wischik:explicit-fusions}.
The syntax used there differs a little from that used in the examples
above: most notably, a prefix written $\op{a}\vec{x}\sdot P$ above is instead
written
$\op{a}\sdot(\langle\vec{x}\rangle\parop P)$ 
 (and symmetrically for inputs); here
$\langle\vec{x}\rangle$ is a vector of datums and the
parallel composition operator is not symmetric for datums. 
Input and output transitions in are on the form $P\gt{a}P'$ where $P'$ is on the form $(\nu\vec{c})(\langle\vec{x}\rangle\parop P)$ and $\vec{c}\subseteq\vec{x}$.
For ease of reading, we write $\langle\vec{x}\rangle P$ for $(\langle\vec{x}\rangle\parop P)$ below.

The encoding of pi-F processes into \textbf{Fi} is as follows:
\[
\begin{array}{rcl}
\encodefi{a\sdot\langle\vec{b}\rangle P} &= & a(\vec{c})\sdot(\pass{\{\vec{b}=\vec{c}\}}\parop \encodefi{P})\quad\text{where }\vec{c}\freshin a\vec{b}\sdot P\\
\encodefi{\op{a}\sdot\langle\vec{c}\rangle P} &= & \op{a}\vec{c}\sdot \encodefi{P}\\
\encodefi{\fusion{\vec{x}}{\vec{y}}} &=&\pass{\{\vec{x}=\vec{y}\}}\\
\end{array}
\]
and is homomorphic for the other operators.  To encode e.g. $\op{a}\sdot(\nu c)\langle c\rangle P$ we first rewrite it to the structurally congruent process $(\nu c)\op{a}\sdot\langle c\rangle P$ (where $c\neq a$).

In~\cite{gardner.wischik:explicit-fusions}, two labelled transition semantics are defined for pi-F and proved to coincide: the \emph{quotiented} and the \emph{structured} semantics.
The first has a traditional rule for using structural congruence ($\equiv$) to derive transitions: if $Q\equiv P\gta P'\equiv Q'$ then $Q\gta Q'$.  
The second semantics has a similar rule but which only allows $\equiv$ to be used after the transition: if $P\gta P'\equiv P_1'$ then $P\gta P_1'$. %
In psi-calculi there is no such structural rule.
For the operational correspondence, however, by the lemma below we can
select a suitable structural representative of the pi-F process.

\begin{lemma}\label{lmm:st-deduction}
In the quotiented semantics of pi-F, if $P\gt{\alpha}P'$ with a deduction tree of depth $n$, there is a deduction tree for the transition of depth no larger than $n$ which uses structural congruence only in its last deduction, or not at all.
\begin{proof}
By induction over $n$.
\end{proof}
\end{lemma}

\newif\iflongpif\longpiftrue %

In the proof below, we make use of the fact that weakening holds in \textbf{Fi}: if $\Psi\vdash\varphi$ then $\Psi\ftimes\Psi'\vdash\varphi$, and thus in particular $\unit\frames P\gt{\alpha}P'$ implies $\Psi\frames P\gt{\alpha}P'$.

\begin{proposition}
In the quotiented semantics of pi-F,
\begin{enumerate}[\em(1)]
\item 
If $P\gt{\op{a}}(\nu\vec{c})\langle\vec{x}\rangle P'$ with $\vec{c}\subseteq\vec{x}$ and $a\freshin\vec{c}$, then there exists a $Q$ s.t. $Q\equiv P$ and
$\unit\frames\encodefi{Q}\gt{\op{a}(\nu\vec{c})\vec{x}} Q'$ and
$\exists P'': P'\equiv P''$ and $Q'=\encodefi{P''}$.
\item 
If $P\gt{a}(\nu\vec{c})\langle\vec{x}\rangle P'$ with $\vec{c}\subseteq\vec{x}$ and $a\freshin\vec{c}$, then there exists a $Q$ s.t. $Q\equiv P$ and
$\unit\frames\encodefi{Q}\gt{a\vec{y}}(\nu\vec{c})(\pass{\{\vec{x}=\vec{y}\}}\parop Q')$ and
$\exists P'': P'\equiv P''$ and $Q'=\encodefi{P''}$.
\item 
If $P\gt{\tau}P'$  then there exists a $Q$ s.t. $Q\equiv P$ and $\unit\frames\encodefi{Q}\gt{\tau}Q'$ and
$\exists P'': P'\equiv P''$ and $Q'=\encodefi{P''}$.
\end{enumerate}

\proof
By Lemma~\ref{lmm:st-deduction}, without loss of generality %
we can assume that the transition of $P$ in the premise can be deduced also for $Q$ without using the transition rule for structural congruence. The proof is then by induction on the depth of the deduction, matching each operational rule of pi-F with a rule in psi.
\iflongpif
\begin{enumerate}[(1)]
\item 
Base case: $P=\op{a}\sdot P_1$ and $P\gt{\op{a}}P_1$ where $P_1=(\nu\vec{c})\langle\vec{x}\rangle P'$ with $\vec{c}\subseteq\vec{x}$. 
We proceed by induction over the length of $\vec{c}$. 
The base case is when $P_1=\langle\vec{x}\rangle P'$, and $\encodefi{P}=\op{a}\vec{x}\sdot\encodefi{P'}$. 
Then $\unit\frames\encodefi{P}\gt{\op{a}\vec{x}}\encodefi{P'}$. 
In the induction case, $P\equiv(\nu\vec{c})\op{a}\sdot\langle\vec{x}\rangle P'=Q$ with $a\freshin\vec{c}$. 
Then $\encodefi{Q}=(\nu\vec{c})\op{a}\vec{x}\sdot\encodefi{P'}$ and by a sufficient number of uses of \textsc{Open}, $\unit\frames\encodefi{Q}\gt{\op{a}(\nu\vec{c})\vec{x}}\encodefi{P'}$.

Induction: we show the case for the parallel rule.  
Here $P=P_1\parop P_2$ and $P_1\gt{\op{a}}(\nu\vec{c})\langle\vec{x}\rangle P_1'$, so $P_1\parop P_2\gt{\op{a}}(\nu\vec{c})\langle\vec{x}\rangle (P_1'\parop P_2)$ with $\vec{c}\freshin P_2$.
By induction, $\unit\frames\encodefi{P_1}\gt{\op{a}(\nu\vec{c})\vec{x}}Q_1'$, and by \textsc{Par} (and weakening) also $\unit\frames\encodefi{P_1}\parop\encodefi{P_2}\gt{\op{a}(\nu\vec{c})\vec{x}}Q_1'\parop\encodefi{P_2}$, since $\vec{c}\freshin\encodefi{P_2}$.

\item 
Similar to the output case, using \textsc{Scope} instead of \textsc{Open} for the induction over $\vec{c}$. 

\item 
Base case: $P=\op{a}\sdot P_1\parop a\sdot P_2$, where $P_i\equiv(\nu\vec{c_i})\langle\vec{x_i}\rangle P_i'$ with $\vec{c_i}\subseteq\vec{x_i}$, for $i\in\{1,2\}$, and $\vec{c_1}\freshin\langle\vec{x_2}\rangle P_2'$ and vice versa. 
Then $P\gtt %
(\nu\vec{c_1}\vec{c_2})(\fusion{\vec{x_1}}{\vec{x_2}}\parop P_1'\parop P_2')$.
By induction, $\unit\frames\encodefi{\op{a}\sdot P_1}\gt{\op{a}(\nu\vec{c_1})\vec{x_1}}Q_1'$ and $\unit\frames\encodefi{a\sdot P_2}\gt{a\vec{x_1}}(\nu\vec{c_2})(\pass{\{\vec{x_1}=\vec{x_2}\}}\parop Q_2')$, 
and $\unit\frames\encodefi{P_1}\parop\encodefi{P_2}\gtt(\nu\vec{c_1}\vec{c_2})(\pass{\{\vec{x_1}=\vec{x_2}\}}\parop Q_1'\parop Q_2')$.\medskip

Induction case: straight-forward, using corresponding operational rules.\qed
\end{enumerate}

\end{proposition}

For the correspondence in the other direction, we use the structured semantics
of pi-F, which has a rule to rewrite labels: if $P\gt{\alpha} P'$ and $P\vdash
\alpha=\beta$, then $P\gt{\beta}P'$, where $P\vdash\varphi$ if the pi-F
correspondent to the frame of $P$ entails $\varphi$.  This is similar to the
rewriting done in psi-calculi, in the prefix and communication rules.

A transition in psi uses an assertion $\Psi$, which needs to be part of the process in pi-F; below we write $\unencodefi{\Psi}$ for the obvious mapping from \textbf{Fi} assertions to pi-F fusions.
In the proofs below, we use results from Section~\ref{sec:properties}, and write $P\equiv Q$ for two psi-calculus agents if they can be proved equal by only Theorems~\ref{thm:fullcong} and~\ref{thm:struct}, which correspond to the standard structural congruence.

\begin{proposition}\hfill
\begin{enumerate}[\em(1)]
\item \label{prop:fi-out}
If $\Psi\frames\encodefi{P}\gt{\op{a}(\nu\vec{c})\vec{x}} P'$ then %
$\unencodefi{\Psi}\parop P\gt{\op{a}}(\nu\vec{c})\langle\vec{x}\rangle Q$ and $\exists Q':Q\equiv\unencodefi{\Psi}\parop Q'$ and $P'=\encodefi{Q'}$
\item \label{prop:fi-in}
If $\Psi\frames\encodefi{P}\gt{a\vec{x}}\equiv(\nu\vec{c})(\pass{\{\vec{x}=\vec{y}\}}\parop P')$ where $\vec{c}\subseteq\vec{y}$ then %
$\unencodefi{\Psi}\parop P\gt{a}(\nu\vec{c})\langle\vec{y}\rangle Q$ and $\exists Q':Q\equiv\unencodefi{\Psi}\parop Q'$ and $P'=\encodefi{Q'}$
\item \label{prop:fi-tau}
If $\Psi\frames\encodefi{P}\gt{\tau}P'$ then %
$\unencodefi{\Psi}\parop P\gt{\tau} Q$ and $\exists Q': Q\equiv \unencodefi{\Psi}\parop Q'$ and $P'=\encodefi{Q'}$.
\end{enumerate}

\proof
By induction over the derivation of the psi-calculus transition.
\iflongpif %
We sometimes use \cite[Lemma 11]{gardner.wischik:explicit-fusions} to restructure a pi-F agent before the transition ($P\equiv P_1\gt{\alpha}P'$ implies $P\gt{\alpha}P'$), and idempotence of fusions.
\begin{enumerate}[(1)]
\item 
Base case: $\Psi\frames\encodefi{P}\gt{\op{a}\vec{x}}P'$ by the \textsc{Out} rule. Then
$P=\op{b}\vec{x}\sdot Q$ and $\Psi\frames\encodefi{P}\gt{\op{a}\vec{x}}\encodefi{Q}$ where $\Psi\vdash a\sch b$, thus in pi-F $\unencodefi{\Psi}\parop P\gt{b}\langle\vec{x}\rangle Q$ and %
 $\unencodefi{\Psi}\parop P\vdash a=b$ so $\unencodefi{\Psi}\parop P\gt{a}\langle\vec{x}\rangle Q$.

Induction: we show the case for \textsc{Open}. Here $\encodefi{P}=(\nu c)\encodefi{P_1}$ and $\Psi\frames\encodefi{P_1}\gt{\op{a}}(\nu\vec{e})P'$ s.t. $c\freshin\vec{e},\Psi,a$ and $c\in\n(\vec{x})$, and by \textsc{Open} $\Psi\frames(\nu c)\encodefi{P_1}\gt{\op{a}}(\nu c\vec{e})P'$.
By induction $\unencodefi{\Psi}\parop P_1\gt{\op{a}}(\nu\vec{e})Q$, and thus $(\nu c)(\unencodefi{\Psi}\parop P_1)\gt{\op{a}}(\nu c\vec{e})P'$, and also $\unencodefi{\Psi}\parop (\nu c)(P_1)\gt{\op{a}}(\nu c\vec{e})P'$.

\item 
Base case: $\Psi\frames\encodefi{P}\gt{a\vec{x}}P'$ by the \textsc{In} rule. Then
$P=b\vec{y}\sdot Q$ and $\Psi\frames\encodefi{P}\gt{a\vec{x}}\pass{\{\vec{x}=\vec{y}\}}\parop\encodefi{Q}$ where $\Psi\vdash a\sch b$, thus in pi-F $\unencodefi{\Psi}\parop P\gt{b}\langle\vec{y}\rangle(\unencodefi{\Psi}\parop Q)$ and as above equally $\unencodefi{\Psi}\parop P\gt{a}\langle\vec{y}\rangle (\unencodefi{\Psi}\parop Q)$.

Induction: we show the case for \textsc{Par}. 
Here $\encodefi{P}=\encodefi{P_1}\parop\encodefi{P_2}$ and 
$\Psi\ftimes\frass{\encodefi{P_2}}\frames\encodefi{P_1}\gt{a\vec{x}}\equiv(\nu\vec{c})(\pass{\{\vec{x}=\vec{y}\}}\parop P_1')$ and by induction
$\unencodefi{\Psi\ftimes\unencodefi{\frass{\encodefi{P_2}}}}\parop P_1\gt{a}(\nu\vec{c})\langle\vec{y}\rangle Q_1$, 
where $\vec{c}\subseteq\vec{y}$ and $\exists Q':Q\equiv(\unencodefi{\Psi}\parop Q')\wedge P_1'=\encodefi{Q'}$.
Then $\Psi\frames\encodefi{P_1}\parop\encodefi{P_2}\gt{a\vec{x}}\equiv(\nu\vec{c})(\pass{\{\vec{x}=\vec{y}\}}\parop P_1')\parop\encodefi{P_2}\equiv(\nu\vec{c})(\pass{\{\vec{x}=\vec{y}\}}\parop P_1'\parop\encodefi{P_2})$ where $\vec{c}\freshin\encodefi{P_2}$. 
In pi-F, we have $\unencodefi{\Psi}\parop\unencodefi{\frass{\encodefi{P_2}}}\parop P_1\parop P_2\gt{a}(\nu\vec{c})\langle\vec{y}\rangle(Q_1\parop P_2)$.
W.l.o.g. (see~\cite[p613]{gardner.wischik:explicit-fusions}) we can assume that $P_2$ is on the form $\phi\parop P''$ where $\phi$ is a fusion and $P''$ has no top-level fusions. Thus $\frnames{\encodefi{P_2}}=\epsilon$, and 
by idempotence of fusions, equally
$\unencodefi{\Psi}\parop\parop P_1\parop P_2\gt{a}(\nu\vec{c})\langle\vec{y}\rangle(Q_1\parop P_2)$.

\item 
Base case: $\Psi\frames\encodefi{P}\gtt Q$  by the \textsc{Com} rule. Then
 $\encodefi{P}=\encodefi{P_1}\parop\encodefi{P_2}$ and
 $\Psi\ftimes\frass{\encodefi{P_2}}\frames \encodefi{P_1}\gt{\op{a}(\nu \vec{c})\vec{x}}P_1'$,
 $\Psi\ftimes\frass{\encodefi{P_1}}\frames \encodefi{P_2}\gt{b\vec{x}}(\nu\vec{e})(\pass{\{\vec{x}=\vec{y}\}}\parop P_2'')$,
 where $\vec{c}\freshin P_2$, $\vec{e}\freshin P_1$, $\vec{e}\subseteq\vec{y}$, and
 $\Psi\ftimes\frass{\encodefi{P_1}}\ftimes\frass{\encodefi{P_2}}\vdash a\sch b$,
as well as
\[\quad\enspace\Psi\frames\encodefi{P_1}\parop\encodefi{P_2}\gt{\tau}(\nu\vec{c})(P_1'\parop(\nu\vec{e})(\pass{\{\vec{x}=\vec{y}\}}\parop P_2''))\equiv(\nu\vec{c}\vec{e})(\pass{\{\vec{x}=\vec{y}\}}\parop P_1'\parop P_2'').\]
Application of (\ref{prop:fi-out}) yields $\unencodefi{\Psi\ftimes\frass{\encodefi{P_2}}}\parop P_1\gt{\op{a}}(\nu\vec{c})\langle\vec{x}\rangle P_1'\parop\unencodefi{\Psi\ftimes\frass{\encodefi{P_2}}}$,
and by (\ref{prop:fi-in}) we obtain
$\unencodefi{\Psi\ftimes\frass{\encodefi{P_1}}}\parop P_2\gt{a}(\nu\vec{e})\langle\vec{y}\rangle P_2''\parop\unencodefi{\Psi\ftimes\frass{\encodefi{P_2}}}$. 

Then %
\[
\begin{array}{rl}
\multicolumn{2}{l}{%
\unencodefi{\Psi\ftimes\frass{\encodefi{P_1}}}\parop\unencodefi{\Psi\ftimes\frass{\encodefi{P_2}}}\parop P_1\parop P_2}\\
\gt{?a=b} & (\nu\vec{c}\vec{e})(\fusion{\vec{x}}{\vec{y}}\parop P_1'\parop P_2''\parop \unencodefi{\Psi\ftimes\frass{\encodefi{P_1}}}\parop\unencodefi{\Psi\ftimes\frass{\encodefi{P_2}}})\\
\equiv & (\nu\vec{c}\vec{e})(\fusion{\vec{x}}{\vec{y}}\parop P_1'\parop P_2''\parop \unencodefi{\Psi\ftimes\frass{\encodefi{P_1}}\ftimes\frass{\encodefi{P_2}}})
\end{array}
\] 
and since 
$\Psi\ftimes\frass{\encodefi{P_1}}\ftimes\frass{\encodefi{P_2}}\vdash a\sch b$ %
the label can be rewritten to $?a=a$ and then %
further rewritten to $\tau$.
As above we can assume that $P$ is on the form $\phi\parop P''$ where $\phi$ is a fusion and $P''$ has no top-level fusions. Thus $\frnames{P_1}=\frnames{P_2}=\epsilon$, and 
by idempotence of fusions, equally
$\unencodefi{\Psi}\parop P_1\parop P_2\gt{\tau}(\nu\vec{c}\vec{e})(\fusion{\vec{x}}{\vec{y}}\parop P_1'\parop P_2'')\parop \unencodefi{\Psi}$.\medskip

Induction: straight-forward matching of transitions rules.\qed
\end{enumerate}
\fi %
\end{proposition}

\subsubsection{Concurrent constraints}

Process calculi which integrate communication and mobility with concurrent constraint (CC) programming have appeared 
e.g.~in~\cite{smolka:foundation-higher-order,niehren.muller:constraints-free,diaz.rueda.ea:pi-calculus,buscemi.montanari:open-bisimulation}.
Here, the \textbf{ask} and \textbf{tell} operations interact with a constraint store. The $\textbf{ask }\varphi\sdot P$ operation checks whether a constraint $\varphi$ is satisfied by the current store and only then proceeds as $P$, corresponding to $\ifthen{\varphi}{\tau\sdot P}$ in psi-calculi.
The $\textbf{tell }\Psi\sdot P$ operation adds a constraint $\Psi$ to the current store before proceeding as $P$. 
Two variants of tell have been identified and used: one which can only proceed if the resulting store is consistent is known as \emph{atomic} tell, and one which allows an inconsistent store and is called \emph{eventual} tell~\cite{saraswat:concurrent-constraint-programming}. 
The eventual tell operation is used in earlier process calculi which integrate constraints and communication, e.g.{} the $\pi^+$-calculus~\cite{diaz.rueda.ea:pi-calculus} and the $\rho$-calculus~\cite{niehren.muller:constraints-free}. 
The atomic tell operation is used in the CC-Pi calculus~\cite{buscemi.montanari:open-bisimulation}.

We here present a psi-calculus with concurrent constraints.  
Similarly to CC-Pi we
extend a basic pi-F-like calculus %
with ask and tell operations and
 use a \emph{named c-semiring} \cite{buscemi.montanari:open-bisimulation} as the constraint system parameter.
Such a constraint system contains names, name fusion/equality constraints and a name hiding operator $\nu$, and supports general
constraint semirings, e.g.{} Herbrand constraints.

Our psi-calculus, call it \textbf{Ci}, with associated named c-semiring $\mathcal{C}=\langle A,\oplus,\otimes,0,1\rangle$ and induced preorder $\preceq$ is:
\begin{mathpar}
\begin{array}{rcl}
\terms &\defn & \nameset\\
\assertions\defn\conditions & \defn & A \\
\sch &\defn& =\\
\ftimes & \defn & \mbox{The similarly notated operator $\otimes$ in $\mathcal{C}$} \\
\emptyframe & \defn &1 \\
\vdash&\defn & \preceq\\
\end{array}
\end{mathpar}

\noindent Thus terms are names, while conditions and assertions are defined by the carrier $A$
of the named c-semiring, which by definition includes names and name fusions, and implicitly name equality conditions.
The properties of named c-semirings guarantee the requirements of psi-calculi, assuming that substitution on the named c-semiring satisfies our requisites.
Abelian monoid properties follow directly, compositionality from $\Psi_1\sequivalent \Psi_2 \Rightarrow \Psi_1=\Psi_2$, and the channel equivalence properties from the fact that = is an equivalence.
We extend the encoding of (monadic) pi-F processes and
represent $\textbf{ask } \varphi\sdot P$ as $\ifthen{\varphi}{\tau\sdot P}$.
An eventual tell operation
$\textbf{tell$_e$ }\Psi\sdot P$ can be represented as $\tau\sdot(\pass{\Psi} \parop P)$.
The atomic \textbf{tell$_a$} operation can be handled by adding a condition $\sym{\consistent}{\Psi}$ to $\conditions{}$ with $\Psi\vdash\sym{\consistent}{\Psi'}$ if $\Psi\ftimes\Psi'$ is consistent, and representing $\textbf{tell$_a$ }\Psi\sdot P$ as $\ifthen{\sym{\consistent}{\Psi}}{\tau\sdot(\pass{\Psi}\parop P)}$.

The most prominent difference from the CC-Pi calculus is that
there, name fusions resulting from communication 
are required to be consistent with the store, otherwise the communication cannot happen. 
In contrast our semantics will allow communication transitions that lead to an inconsistent store. 
This difference is illustrated below:
\[
\begin{array}{rcl}
\multicolumn{3}{@{}l}{\text{In CC-Pi:}}\\
P&=&\Psi\parop\op{a}{b}\sdot Q\parop {c}{d}\sdot R
 \gtt
 \Psi\ftimes (b=d)\parop Q\parop R \\
&&\multicolumn{1}{r}{
\text{if $\Psi\preceq a=c$ and $\Psi\ftimes (b=d)$ consistent} }
\\
\multicolumn{3}{@{}l}{\text{In \textbf{Ci}:}}\\
P &=&
\pass{\Psi}\parop \out{a}{b}\sdot Q\parop {c}{(x)}\sdot(\pass{x=d}\parop R)\\
\multicolumn{2}{r}{\gtt}&
 \pass{\Psi}\parop \pass{b=d}\parop Q\parop R
\qquad
\text{if $\Psi\vdash a=c$}
\end{array}
\]
While it appears not possible to integrate an atomic consistency check in a psi-calculus communication without changing our \textsc{Com} rule,
explicit consistency checks (like $\sym{cons}{\Psi}$) can be used to handle interesting applications in practice.
The semantics of CC-Pi is given by a structural congruence and a reduction relation. There is also a labelled operational semantics, but it is in fact not compositional. Consider the CC-Pi agents
\[P=\res{b,x}(x=b\parop \out{a}{x}\sdot b\sdot c) \qquad
Q=\res{b,x}(x=b\parop \out{a}{x})\]
 (where insignificant objects are omitted). They have the same constraint store and the same transitions in all constraint contexts.
However, they do not have the same transitions in all process contexts:
a parallel context
$R=a(y).\op{y}$ tells the difference:
\[P\parop R\gtt\gtt \res{b}(x=b\parop x=y\parop c)
\gt{c}\]
while $Q\parop R$ of course has no such $c$ transition. Thus Theorem~1
of~\cite{buscemi.montanari:open-bisimulation} is incorrect: open
bisimilarity is not a congruence (see also~\cite{buscemi.montanari:ccpi-errata}).  
\label{sect:ccpianomaly}

The labelled semantics of CC-Pi has a curious asymmetry between the
rule for prefixes and the rule for communication: in the first case,
the constraint store cannot affect the label induced by the prefix,
while in the communication case, the constraint store judges whether
the subjects should be considered the same, enabling the
communication.  The psi-calculi have no such asymmetry: the assertions (corresponding to
the store) allow the subject to be rewritten in the prefix rules and
the subjects in \textsc{Com} are compared using the assertions (see 
Section~\ref{sec:channel-equivalence-everywhere} for a discussion).  A possible fix for
CC-Pi would involve allowing the store to rewrite terms, thus also subjects in
prefixes~\cite{buscemi:personal}. 

\bigskip
Psi-calculi go beyond most concurrent constraint systems in two
ways. Firstly, we allow arbitrary logics, even higher-order
ones. Secondly, we allow constraints and conditions to be data terms, which means 
an agent can transmit and receive these. 
For example, assume that \texttt{c} is a constraint and that
\texttt{f} is a function from assertions to assertions. Then in the agent
$\overline{a}\,{\tt c} \sdot P \parop a(z) \sdot (\pass{{\tt f}(z)} \parop Q)\gtt P\parop \pass{\mathtt{f}(\mathtt{c})}\parop Q$
the left hand agent sends the constraint {\tt c} to the right, and {\tt f} is applied to it. 
Similarly, if \texttt{p} is a unary predicate, in the agent
$\out{a}\,\mathtt{p}\sdot P \parop a(z)\sdot \ifthen{z(x)}{Q}\gtt P\parop \ifthen{\mathtt{p}(x)}{Q}$
the left hand agent sends the predicate to the right, which applies it to $x$.

\section{Applications}
\label{sec:applications}
In this section we will look at a few applications of \psicalculi{}, some of which have
been described before in other formalisms, and some which are novel.
\subsection{Structured terms as channels}
Calculi with channels that can carry complex data are common, but in most cases
the terms that represent channels are very simple,
usually only a single name. We here give some examples where they have
structure, and thus may contain more than one name.
\subsubsection{Frequency hopping spread spectrum}
\label{sect:frequency-hopping}
Wireless communication over a constant radio frequency has a number of
drawbacks. In a hostile environment a radio can be tuned in to the correct
frequency and
monitor the communication which is also vulnerable to
jamming.
A solution to these
problems is to jump
quickly between different frequencies in a scheme called frequency hopping
spread spectrum (FHSS), first patented in 1942 \cite{PATENT:US2292387}. 
To eavesdrop it would then be necessary to match both the order of the
frequencies and the pace of switching. Jamming is also made more difficult since
the available power would have
to be distributed over many frequencies. 

We will here show how this is modelled
in a psi-calculus. It is assumed that the initiator of the
communication and the
receiver share an algorithm used to calculate the next frequency.
The procedure starts by the initiator sending a communication request over some
predetermined frequency. The receiver then sends a seed back to the initiator
and both use it to calculate the sequence of frequencies to be used. Then the
initiator synchronises over the first calculated frequency to verify that it got
the right sequence. The communication then proceeds and
both parties change frequencies accordingly.

We will now look at the \psicalculus{} used to model this frequency hopping
algorithm. We let terms represent radio frequencies and use the unary
function $\sym{nextFreq}{M}$ to represent the
algorithm for calculating the next frequency, given the previous frequency $M$.
This \psicalculus{} has no assertions other than unit.
\begin{mathpar}
\begin{array}{rcl}
\terms & \defn & \nameset \cup \{\sym{nextFreq}{M}: M\in\terms\}\\
\conditions & \defn & \{M \sch N : M, N \in \terms\}\\
\assertions & \defn & \{\emptyframe\}\\
\ftimes & \defn & \lambda \Psi_1, \Psi_2.\; \emptyframe \\
\multicolumn{3}{l}{\emptyframe\vdash M\sch M\quad}\\
\end{array}
\end{mathpar}
We define $\top$ to be $a\sch a$ in order to be able to use non-deterministic
choice as noted in Section~\ref{sec:agents}.

Let $X_{\mathit{in},\mathit{out}}$ be an arbitrary agent
that communicates with the environment via the channels $\mathit{in}$ and
$\mathit{out}$. This
agent will be wrapped in contexts that will let it do FHSS in a transparent way:
from the agent's point of view it will
only communicate over the local channels $in$ and $\mathit{out}$. The
agent
$\mathit{FHSS}$
that implements frequency hopping looks like:
\[
 \mathit{FHSS} =\; \bang \inprefixempty{\mathit{fh}}(\mathit{freq})\sdot\left(
\begin{array}{@{}l@{}}
 \inprefixempty{\mathit{out}}(y)\sdot\opa{\mathit{freq}}{y}
\sdot\opa{\mathit
{fh}}{\sym{nextFreq}{\mathit{freq}}} \\
 +\\
 \inprefixempty{\mathit{freq}}(y)\sdot\opa{in}{
y}\sdot\opa{\mathit{fh}}{\sym{nextFreq}{\mathit{freq}}}
\end{array}
\right)
\]
This agent can be thought of as a function $\mathit{fh}$ that will take a
frequency
and then either wait for something to be received from the local channel
$\mathit{out}$
to
send over this frequency, or to receive something over this frequency and
forward it to the local channel $in$. It will then calculate the next frequency
and
start over.

The behaviour when the agent $X_{\mathit{in},\mathit{out}}$ acts
as initiator is
modelled as a context where the initiating sequence starts by
sending a synchronisation message $\mathit{sync}$ over a predetermined control
channel $ctl$, and
then waits for a seed from that channel. It then starts the frequency
hopping algorithm with the seed and sends a synchronisation message over the
first
frequency, and behaves as $X_{\mathit{in},\mathit{out}}$. It is
assumed that
$\mathit{seed} \freshin X_{\mathit{in},\mathit{out}}$.
\[
 I[X_{\mathit{in},\mathit{out}}] =
\opa{\mathit{ctl}}{\mathit{sync}}\sdot\inprefixempty{
\mathit{ctl}}
(\mathit{seed})\sdot\opa{\mathit{fh}}{
\mathit{seed}}\sdot\opa{\mathit{out}}{
\mathit{sync}}\sdot
X_{\mathit{in},\mathit{out}} \parop \mathit{FHSS} 
\]

The behaviour when the agent $X_{\mathit{in},\mathit{out}}$ acts as a
receiver is
modelled similarly: the receiver listens to the control channel
$\mathit{ctl}$ and sends back a seed. Then it starts the frequency hopping
algorithm with this seed and waits for a
synchronisation message. The receiver then behaves as
$X_{\mathit{in},\mathit{out}}$. It is
assumed that $x,\mathit{seed},s \freshin X_{\mathit{in},\mathit{out}}$.
\[
 R[X_{\mathit{in},\mathit{out}}] = 
\inprefixempty{\mathit{ctl}}(s)\sdot
(\nu \mathit{seed})\opa{\mathit{ctl}}{
\mathit{seed}}\sdot\opa{\mathit{fh}}{\mathit{seed}
}\sdot\inprefixempty{in}(x)\sdot X_{\mathit{in},\mathit{out}} 
\parop \mathit{FHSS}
\]

The full system where $X_{\mathit{in},\mathit{out}}$ may behave as either a
receiver
or initiator is then modelled as
\[
 \mathit{FH}[X_{\mathit{in},\mathit{out}}] = (\nu
\mathit{fh},\mathit{in},\mathit{out})\left( I[X_{\mathit{in},\mathit{out}}] +
R[X_{\mathit{in},\mathit{out}}] \right)
\]
where it is assumed that $\mathit{fh} \freshin X_{\mathit{in},\mathit{out}}$.

Let us look at a few transitions of the receiver. First the receiver gets a
request to do frequency hopping over the control channel:
\[\begin{array}{l}
 \mathit{FH}[X_{\mathit{in},\mathit{out}}] 
\gt{\inn{\mathit{ctl}}\mathit{sync}}\\
\qquad(\nu
\mathit{fh},\mathit{in},\mathit{out})\left((\nu
\mathit{seed})\opa{\mathit{ctl}}{
\mathit{seed}
}\sdot\opa{\mathit{fh}}{\mathit{seed}}\sdot\inprefixempty{in}(x)\sdot
X_{\mathit{in},\mathit{out}} 
\parop \mathit{FHSS} \right)
\end{array}
\]
It then sends the seed to the initiator and starts the frequency hopping using
this seed:
\[\begin{array}{l}
\trans{}{\bout{\mathit{seed}}{\mathit{ctl}}{\langle
\mathit{seed}
\rangle}}{(\nu
\mathit{fh}, \mathit{in},\mathit{out})\left(\opa{\mathit{fh}}{\mathit{seed}
}\sdot\inprefixempty{in}(x)\sdot X_{\mathit{in},\mathit{out}} 
\parop \mathit{FHSS} \right)}\\
\trans{}{\tau}{(\nu
\mathit{fh}, \mathit{in},\mathit{out})
\left(\begin{array}{@{}l@{\hspace{1ex}}l@{}}
&\inprefixempty{in}(x)\sdot X_{\mathit{in},\mathit{out}} \\[0.6em]
\big|&
\left(
\begin{array}{@{}l@{}}
 \inprefixempty{\mathit{out}}(y)\sdot\opa{\mathit{seed}}{
y}\sdot\opa{\mathit{fh}}{\sym{nextFreq}{\mathit{seed}}} \\
 +\\
 \inprefixempty{\mathit{seed}}(y)\sdot\opa{\mathit{in}}{
y}\sdot\opa{\mathit{fh}}{\sym{nextFreq}{\mathit{seed}}}
\end{array}\right) \pll \mathit{FHSS}\end{array}\right)}
  \end{array}
\]
At this point the initiator will send the $\mathit{sync}$ message:
\[\begin{array}{rl}
\multicolumn{2}{l}{\gt{\inn{\mathit{seed}}{
\mathit{sync}}}}\\
&(\nu
\mathit{fh}, \mathit{in},\mathit{out})\left(\inprefixempty{\mathit{in}}(x)\sdot
X_{\mathit{in},\mathit{out}} \pll 
 \opa{\mathit{in}}{\mathit{sync}}\sdot\opa{\mathit{fh}}{
\sym{nextFreq}{\mathit{seed}}}
 \pll \mathit{FHSS}\right)\\
\gt{\tau} & (\nu
\mathit{fh}, \mathit{in},\mathit{out})\left(X_{\mathit{in},\mathit{out}} \pll
\opa{\mathit{fh}}{
\sym{nextFreq}{\mathit{seed}}} \pll \mathit{FHSS}\right)
  \end{array}
\]
After another $\tau$-transition the agent is ready to communicate over the
next frequency:
\[
 \begin{array}{@{}rl@{}}
\trans{}{\tau}{(\nu
\mathit{fh}, \mathit{in},\mathit{out})
\left(\begin{array}{@{}l@{\hspace{1ex}}l@{}}
&X_{\mathit{in},\mathit{out}} \\[0.6em]
\big|&
\left(
\begin{array}{@{}l@{}}
 \inprefixempty{\mathit{out}}(y)\sdot\opa{\sym{nextFreq}{\mathit{seed}}}{
y}\sdot\opa{\mathit{fh}}{\sym{nextFreq}{\sym{nextFreq}{\mathit{seed}}}} \\
 +\\
 \inprefixempty{\sym{nextFreq}{\mathit{seed}}}(y)\sdot\opa{\mathit{in}}{
y}\sdot\opa{\mathit{fh}}{\sym{nextFreq}{\sym{nextFreq}{\mathit{seed}}}}
\end{array}\right) \pll \mathit{FHSS}\end{array}\right)}
  \end{array}
\]

This example could easily be made more complex by adding relevant error checking
(e.g. the receiver could check that the synchronisation message is correct), but
even in this form it illustrates the use of structured channels.

\subsubsection{Local services}
A common scenario is that different servers implement the same kind of
functionality known under some globally known name. HTTP servers are
examples of this where the service provided is normally available on IP port 80.
Here the name of the service (port 80) is shared among the different servers.
The general problem is that there is a service known under
a global name, but available from servers with different
names. This problem is treated in depth in
\cite{chothia.stark:distributed-calculus} where the authors invent a new
calculus for this purpose. Here we show how the same problem can be solved
using an instance of \psicalculi{}.

The instance used is basically the same as for polyadic pi-calculus as
presented in
Section~\ref{sec:polyadic-pi-instance} augmented with
terms of form $\locchan{M}{N}$ and the entailment $\unit \vdash
\locchan{M}{N} \sch \locchan{M}{N}$, where $M$ and $N$ are terms. This gives the
possibility to
scope a part of a channel term, e.g $(\nu b)(\out{\locchan{a}{b}}{c}.P)$, as
in~\cite{carbone.maffeis:expressive-power}.
\begin{mathpar}
\begin{array}{rcl}
\terms & \defn & \{\locchan{M}{N} : M,N \in \terms\} \cup\\
       & & \{\mathrm{t_2}(M,N) : M,N \in \terms\} \cup \nameset\\
\conditions & \defn & \{M \sch N : M, N \in \terms\}\\
\assertions & \defn & \{\emptyframe\}\\
\ftimes & \defn & \lambda \Psi_1, \Psi_2.\; \emptyframe \\
\multicolumn{3}{l}{\emptyframe\vdash a\sch a\quad\forall a\in\nameset}\\
\multicolumn{3}{l}{\emptyframe\vdash \locchan{M}{N} \sch
\locchan{M}{N}\quad\forall M,N\in\terms}
\end{array}
\end{mathpar}

The following example is adapted from
\cite{chothia.stark:distributed-calculus}. Assume there are globally
known names
$\mathit{finger}$ and $\mathit{daytime}$ which refer to resources
located at some server. Different servers have different local information, but
this information is accessed through the same globally known names. This can
be modelled as
\[
\begin{array}{rcl}
 \mathit{Server} &=&\bang
\inprefixempty{\mathit{server}}(\mathrm{t_2}(\mathit{service},\mathit{replyc}))
\sdot(\nu a)\left(
   \begin{array}{rl}
        & \opa{\locchan{\mathit{service}}{a}}{\mathit{replyc}}\sdot\nil\\
 \parop & \mathit{Finger}(a)\\
 \parop & \mathit{Daytime}(a)
   \end{array} \right)\\
 \mathit{Finger}(a) &=&
\inprefixempty{\locchan{\mathit{finger}}{a}}(\mathit{replyc})
\sdot\opa{\mathit{replyc}}{\mathit{UserList}}\sdot\nil\\
 \mathit{Daytime}(a) &=&
\inprefixempty{\locchan{\mathit{daytime}}{a}}(\mathit{replyc})
\sdot\opa{\mathit{replyc}}{\mathit{Date}}\sdot\nil\\
\end{array}
\]
where $\mathit{UserList}$ and $\mathit{Date}$ are some terms containing the
requested information. The exact nature of these terms is unimportant for this
example.

The server listens to incoming requests on channel $server$ and receives two
names. The first name is the requested service, and the second is the
reply channel. It will then do an internal communication with the particular
service daemon. There is no risk of interference since a locally scoped name is
part of the service channel. The result of the request is then forwarded along
the reply channel. 

\[\begin{array}{rl}
\mathit{Server}&\trans{}{\inn{\mathit{server}}{\mathrm{t_2}(
\mathit { finger } , c ) } } { \mathit{Server} \pll (\nu a)\left(
   \begin{array}{rl}
        & \opa{\locchan{\mathit{finger}}{a}}{c}\sdot\nil\\
 \parop & \mathit{Finger}(a)\\
 \parop & \mathit{Daytime}(a)
   \end{array} \right)}\\
               &\trans{}{\tau}{\mathit{Server} \pll (\nu a)\left(
   \begin{array}{rl}
         \nil \parop  \opa{\mathit{c}}{\mathit{UserList}
}\sdot\nil  \parop  \mathit{Daytime}(a)
   \end{array} \right)}\\
               &\trans{}{\opa{c}{
\mathit{UserList}}}{\mathit{Server} \pll (\nu a)\left(
   \begin{array}{rl}
        \nil \parop \nil \parop \mathit{Daytime}(a)
   \end{array} \right)}
\end{array}
\]
Since any transitions from $\mathit{Daytime}(a)$ are prevented by the
restriction, the final derivative will behave like $\mathit{Server}$.

\subsection{Cryptography}\label{sec:crypto-example}
In this section we give a sequence of examples from cryptography, culminating
with a model of the Diffie-Hellman key agreement protocol. Our exposition is quite
similar to the applied pi-calculus as presented
in~\cite{abadi.fournet:mobile-values}, and we will use a psi-calculus that mimics
this closely. The main point is that psi-calculi can express these cryptographic
examples in an equally concise way, and within a leaner and more symmetric
formalism. 

The psi-calculus instance we use for the examples below can be seen as a
simplification of \apiinstance{} in Section~\ref{sect:api}
in that we do not distinguish between different kinds of names, and we do not use inequality.
To construct this psi-calculus we assume an inductively defined set of terms
using a signature $\Sigma$, and an equational theory $\vdash_\Sigma$
which let us infer equations $M=N$ where $M$ and $N$ are terms. 
Exactly how this theory works is unimportant for this presentation.
Substitution is defined in the expected way.
\begin{mathpar}
\begin{array}{rcl}
\terms & \defn & \nameset \cup \{\mathsf{f}(M_1,\dots,M_n) :
\mathsf{f} \in \Sigma \wedge M_i \in \terms\}\\
\conditions & \defn & \{M=N: M,N\in\terms\}\\
\assertions & \defn & \mathcal{P}_{\rm fin}(\{M=N: M,N\in\terms\})\} \\
\sch & \defn & =\\
\ftimes & \defn & \cup
\\
\emptyframe & \defn &\emptyset\\
\multicolumn{3}{l}{\Psi\vdash M=N\text{ if }\vdash_{\Sigma\cup\Psi} M =
N}\\
\end{array}
\end{mathpar}
An assertion is a finite set of equations between terms. We often elide the
set brackets in agents, e.g. writing $\pass{M=N}$ instead of $\pass{\{M=N\}}$.
The conditions are just equations $M=N$.
Entailment is defined such that $\Psi \vdash M=N$ holds if  $M=N$ can be
inferred from
the equational theory  $\vdash_\Sigma$ extended by the equations in $\Psi$.
This instance satisfies the requirements by the same reasoning as for
\apiinstance{}.

We start by looking at how one-way hashing is modelled.
In addition to symbols for tupling and projection, and their associated
equations, the signature contains the unary symbol $\sym{hash}{x}$ which has no
equations.
The only equation on $\symempty{hash}$ that is true is
$\sym{hash}{M} = \sym{hash}{M}$, and this means that the hash function is
collision free.
The following example shows
one agent sending a message $M$ together with a hashing $x$ of the message
and
a secret name $s$ to another agent. The second agent will only forward
$M$ if it is properly hashed.
\[
 (\nu s)(\equat{\sym{hash}{\pair{s}{M}}}{x} \parop \opa{a}{
\pair{M}{x}
} \parop \inprefixempty{a}(y)\sdot\ifthen{\sym{hash}{\pair{s}{\fst{y}}}
= \snd{y}}{\opa{b}{\fst{y}}})
\]

To model symmetric cryptography, the signature is extended as in Section~\ref{sect:api}: we add
the binary symbols $\sym{enc}{x,y}$ and $\sym{dec}{x,y}$ 
together with the
equation $\sym{dec}{\sym{enc}{x,y},y} = x$.
 The following agent sends a
message $M$ encrypted with the secret key $k$, without revealing the plaintext or key.
\[
(\nu k,x)(\equat{\sym{enc}{M,k}}{x} \parop \out{a}{x})\gt{\out{a}{(\nu x)x}}
\dots
\]

Asymmetric encryption is modelled by adding two new unary symbols
$\sym{pk}{s}$ and $\sym{sk}{s}$ which generate the public and secret keys from a common seed value,
and the equation
$\sym{dec}{\sym{enc}{x,\sym{pk}{k}},\sym{sk}{k}}=x$. The following
agent
sends the public key on channel $a$, receives a message along channel $b$,
decrypts it with the secret key, and sends the decrypted message along channel
$c$:
\[
 (\nu s,x)(\equat{\sym{pk}{s}}{x} \parop \out{a}{x} \parop
\inprefixempty{b}(y).(\equat{\sym{dec}{y, \sym{sk}{s}}}{z}) \parop \out{c}{z})
\]
Non-deterministic crypto is modelled by using
a ternary version of the symbol $\sym{enc}{x,y,z}$ with some salt in the last
argument, together with the equation
$\sym{dec}{\sym{enc}{x,\sym{pk}{k},z},\sym{sk}{k}}=x$. Consider the following
agent:
\[
 \inprefixempty{a}(x).\big((\nu m,y)(\equat{\sym{enc}{M,x,m}}{y} \parop \out{b}{y})
\parop (\nu n,z)(\equat{\sym{enc}{M,x,n}}{z}) \parop \out{c}{z}\big)
\]
An observer of this agent cannot tell whether $y$ and $z$ are encryptions
of the same message or not, because of the unique salt.

Digital
signatures are modelled by adding the binary symbol $\sym{sign}{x,y}$, the ternary symbol
$\sym{check}{x,y,z}$,
the constant symbol $\symempty{ok}$, and the equation
$\sym{check}{x,\sym{sign}{x,\sym{sk}{k}},\sym{pk}{k}}=\symempty{ok}$. The
following agent sends a signed message along
$a$, then the parallel component receives it and checks the signature. If it is
ok it is then forwarded.
\[
\begin{array}{l@{}l}
(\nu s,z)(&\equat{\sym{pk}{s}}{y} \parop \equat{\sym{sign}{M,
\sym{sk}{s}}}{z} \parop \out{a}{\pair{M}{z}}) \\
&\parop
\inprefixempty{a}(x).\ifthen{\sym{check}{\fst{x},\snd{x},y} =
\symempty{ok}}{\out{b}{\fst{x}}}
\end{array}
\]

The Diffie-Hellman protocol \cite{diffie.hellman:new-directions-in-cryptography}
is used to establish a shared secret between two
principals who do not necessarily share any secrets beforehand. This is done by
exchanging messages over a public channel.

We let $\Sigma$ include $\sym{f}{x,y}$ and $\sym{g}{x}$, and the equation
system includes $\sym{f}{x,\sym{g}{y}} = \sym{f}{y, \sym{g}{x}}$, but no other
equations on $\mathsf{f}$ and $\mathsf{g}$. The first principal $P$ creates a
secret $n_P$ and sends an alias $x_P$ of $\sym{g}{n_P}$ to the other principal $Q$,
and $Q$ does likewise. Then $P$ can create the term 
$\sym{f}{n_P, x_Q}$ and $Q$ can create the term 
$\sym{f}{n_Q, x_P}$. 
Using the equations above these two terms are equivalent and
the shared secret has been established. 
Concretely $\mathsf{f}$ and $\mathsf{g}$
are functions in a multiplicative group modulo a large prime,
 but here we ignore the number theory.

Let $P_{k_P}$ and $Q_{k_Q}$ be two agents that will share a secret key and
will use the names $k_P$ and $k_Q$, respectively, to refer to it. The
Diffie-Hellman key agreement is modelled as two symmetric contexts
$\mathit{DH_{01}}[\cdot]$ and $\mathit{DH_{10}}[\cdot]$
in which the agents are placed. The context $\mathit{DH_{01}}[X_k]$ is defined
as
\[\mathit{DH_{01}}[X_k] = (\nu
n,x,a_{01},a_{10})(\equat{\sym{g}{n}}{x} \parop \out{a_{01}}{x} \parop
\inprefixempty{a_{10}}(z).(\nu k)(\equat{\sym{f}{n,z}}{k}
\parop X_k))\]
where $n,x,a_{01},a_{10} \freshin X_k$ and $k$ occurs in $X_k$ as a a name
that refers to a key. The context $\mathit{DH_{10}}[X_k]$ is defined in the
same way but with $a_{10}$ and $a_{01}$ swapped.

The agents $P_{k_P}$ and $Q_{k_Q}$ agree on the secret by placing them in the
contexts: $\mathit{DH_{01}}[P_{k_P}]$ and
$\mathit{DH_{10}}[Q_{k_Q}]$. The key agreement will
then do two internal transitions:
\[
\trans{\mathit{DH_{01}}[P_{k_P}] \parop
\mathit{DH_{10}}[Q_{k_Q}]}{\tau}{\trans{}{\tau}{(\nu x_P,x_Q)(P'
\parop Q')}}
\]
where
\[
\begin{array}{rcl}
P' &=& (\nu n_P,a_{01},a_{10})(\equat{\sym{g}{n_P}}{x_P} \parop (\nu
k_P)(\equat{\sym{f}{n_P,x_Q}}{k_P} \parop P_{k_P}))\\
Q' &=& (\nu n_Q,a_{01},a_{10})(\equat{\sym{g}{n_Q}}{x_Q} \parop (\nu
k_Q)(\equat{\sym{f}{n_Q,x_P}}{k_Q} \parop Q_{k_Q}))
\end{array}
\]
The $x$ and $n$ from the context have been alpha-converted to the
variants with subscripts to avoid clashes.

Since the agents are communicating over a public channel
the messages may be intercepted by a passive attacker which then forwards them
unmodified. In presence of such an attacker the agents evolve to $P'
\parop Q'$ where the lack of binders for $x_P$ and
$x_Q$ represent that
the hostile environment now has access to these values. We show that this does not break the protocol.

As a specification for this protocol we put $P_{k_P} \parop Q_{k_Q}$ in a
context
where they already share a secret, here represented by the name $k'$: $S = (\nu
k_P, k_Q, k')(\equat{k'}{k_P}
\parop \equat{k'}{k_Q} \parop P_{k_P} \parop Q_{k_Q})$. We
then show that $P' \parop Q'$ and $S$ behave the same, denoted $P'
\parop Q' \sim S$. The precise meaning of $\sim$ is
given in Section~\ref{sec:bisimilarity}, but for this particular example it is
sufficient to think of $\sim$ as equivalence of the frames of $S$ and $P' \pll
Q'$ according to Definition~\ref{def:frame-equivalence}. This equivalence is
closed under parallel composition (if $P$ and $Q$ behave the same, then so will
$P\parop R$ and $Q \parop R$ for any agent $R$) and restriction (if $P$ and $Q$
behave the same, then so will $(\nu a)P$ and $(\nu a)Q$, for any $a$). 

We have that
\[
\begin{array}{@{}r@{}l@{}}
&\begin{array}[t]{@{}l@{}l@{}}
 &(\nu n_P,a_P)(\equat{\sym{g}{n_P}}{x_P} \parop \equat{\sym{f}{n_P,x_Q}}{k_P})\\
\parop& (\nu n_Q,a_Q)(\equat{\sym{g}{n_Q}}{x_Q} \parop \equat{\sym{f}{n_Q,x_P}}{k_Q})\\
\end{array}\\
\sim \\
& (\nu k')(\equat{k'}{k_P} \parop \equat{k'}{k_Q})
\end{array}
\]
The reason is that the only condition entailed on both sides is $k_P = k_Q$, no
equalities can
be entailed on $x_P$ and $x_Q$. Since $\sim$ is closed under parallel
composition we can add the agents: %
\[
\begin{array}{@{}r@{}l@{}}
&\begin{array}[t]{@{}l@{}l@{}}
 &(\nu n_P,a_P)(\equat{\sym{g}{n_P}}{x_P} \parop \equat{\sym{f}{n_P,x_Q}}{k_P})\\
\parop& (\nu n_Q,a_Q)(\equat{\sym{g}{n_Q}}{x_Q} \parop \equat{\sym{f}{n_Q,x_P}}{k_Q})\\
\parop& P_{k_P} \parop Q_{k_Q}\\
\end{array}\\
\sim \\
&\begin{array}[t]{@{}l@{}l@{}}
& (\nu k')(\equat{k'}{k_P} \parop \equat{k'}{k_Q})\\
\parop& P_{k_P}\parop Q_{k_Q}
 \end{array}
\end{array}
\]
Since $\sim$ is closed under the restriction operator: %
\[
\begin{array}{@{}l@{}l@{}}
&(\nu k_P,k_Q)\\
&\quad (\begin{array}[t]{@{}l@{}l@{}}
& (\nu n_P,a_P)(\equat{\sym{g}{n_P}}{x_P} \parop \equat{\sym{f}{n_P,x_Q}}{k_P})\\
\parop& (\nu n_Q,a_Q)(\equat{\sym{g}{n_Q}}{x_Q} \parop \equat{\sym{f}{n_Q,x_P}}{k_Q})\\
\parop& P_{k_P} \parop Q_{k_Q})\\
       \end{array}\\
\sim \\
&(\nu k_P,k_Q)\\
&\quad (\begin{array}[t]{@{}l@{}l@{}}
	&(\nu k')(\equat{k'}{k_P} \parop \equat{k'}{k_Q}) \\
   \parop& P_{k_P}\parop Q_{k_Q})\\
	\end{array}
\end{array}
\]
Finally, by the structural laws of Theorem~\ref{thm:struct} in
Section~{\ref{sec:properties}}:
\[
 P' \parop Q' \sim S.
\]

\section{Bisimilarity}
\label{sec:bisimilarity}
In this section we define a notion of strong bisimilarity on agents and prove that it satisfies the expected algebraic laws and substitutive properties. The results hold for any psi-calculus and give us confidence in the semantic definitions.
\subsection{Definition}

In the standard pi-calculus the notion of strong bisimulation is used to formalise the intuition that two agents ``behave in the same way''; it is defined as a symmetric binary relation $\mathcal{R}$ satisfying the simulation property: $\mathcal{R}(P,Q)$ implies that for $\alpha$ such that $\bn{\alpha} \freshin Q$, 
\[\text{if }\trans{P}{\alpha}{P'}\text{ then } 
\trans{Q}{\alpha}{Q'}
\land  {\mathcal R}( P', Q')\]
For a psi-calculus we additionally need to take the assertions into
consideration. The behaviour of an agent is always taken with respect to an
environmental assertion. We define bisimulation as a ternary relation
$\mathcal{R}(\Psi,P,Q)$, saying that $P$ and $Q$ behave in the same way when the
environment asserts $\Psi$.
Because of this two additional issues arise. The first is that the agents can
affect their environment through their frames (and not only by performing
actions), and this must be represented in the definition of bisimulation. The
second is that the environment (represented by $\Psi$ in
$\mathcal{R}(\Psi,P,Q)$) can change, and for $P$ and $Q$ to be bisimilar they
must continue to be related after such changes. This leads to the following
definition of strong bisimulation.

\begin{definition}[Bisimulation]
A {\em bisimulation}
 $\mathcal{R}$ is a ternary relation between assertions and pairs of agents such that
 ${\mathcal{R}}(\Psi,P,Q)$ implies all of
 \begin{enumerate}[(1)]%
 \item Static equivalence:
   \label{def:case:bisimStatEq}
  $\Psi \ftimes \fr{P} \sequivalent \Psi \ftimes \fr{Q}$
 \item
   Symmetry: ${\mathcal{R}}(\Psi,Q,P)$
 \item
   \label{def:case:bisimExtAss}
 Extension of arbitrary assertion:
 $\forall \Psi'.\ {\mathcal{R}}(\Psi \ftimes \Psi',P,Q)$
 \item   Simulation:
for all $\alpha, P'$ such that
$\bn{\alpha}\freshin \Psi,Q$ there exists a $Q'$ such
that
\[\text{if }\framedtransempty{\Psi}{P}{\alpha}{P'}\text{ then }
\framedtransempty{\Psi}{Q}{\alpha}{Q'}
\land  {\mathcal R}(\Psi , P', Q')\]
\end{enumerate}
 \label{def:bisim}
We define $P \bisim_\Psi Q$ to mean that there exists a bisimulation ${\mathcal{R}}$ such that
${\mathcal{R}}(\Psi,P,Q)$, and write $\bisim$ for $\bisim_\emptyframe$.
\end{definition}
Clauses 2 and 4 are familiar from the pi-calculus. Clause 1 captures the fact
that the related agents have exactly the same influence on the environment
through their frames, namely that when they add to the existing environment
($\Psi$) then exactly the same conditions are entailed. Clause 3 means that when
the environment changes (by adding a new assertion $\Psi'$) the agents are still
related.
An example may clarify the role of this clause. Let  $\beta$ be a prefix and let $\varphi$ be any non-trivial condition, and consider
\[ \begin{array}{l}
P = \beta.\beta.\nil + \beta.\nil + \beta . \;\ifthen{\varphi}{\beta.\nil} \\
Q = \beta.\beta.\nil + \beta.\nil
\end{array}\]
$P$ can non-deterministically choose between three branches and $Q$ between the two first of them. Here
$P$ and $Q$ are not bisimilar. If $P$ performs an action corresponding to its
third case, reaching the agent $P' = \mbox{\bf if}\;\varphi\;\mbox{\bf then}
\;\beta.\nil$, there is no way that $Q$ can simulate since neither $Q' = \nil$ nor
$Q'= \beta.\nil$ is equivalent to $P'$ in all environments. In fact, any reasonable variant
of bisimulation that equates $P$ and $Q$ will not be preserved by
parallel. To see this, let $T$ be $\gamma . \pass{\Psi}$, where $\gamma$ is any
prefix and $\Psi$ an assertion that entails $\varphi$. Then the transition
$\trans{P\parop T}{\beta}{P'\parop T}$ cannot be simulated by $Q|T$, since $P'|T$ can only do
an action $\gamma$ followed by an action $\beta$, whereas $\beta.\nil|T$ can do
$\beta$ immediately, and $\nil|T$ can do no $\beta$ at all. This demonstrates why
clause 3, extension of arbitrary assertion, is necessary: it says that after
each step all possible extensions of the assertion must be considered. If we
would  merely require this at top level, i.e. remove clause 3 and instead
require $\forall \Psi .{\mathcal{R}}(\Psi,P,Q)$ in the definition of $P \bisim Q$,
the extensions would not recur; as a consequence $P$ and $Q$ in the example
would be equivalent, and the  equivalence would not be preserved by parallel.

For another example, consider
\[ \begin{array}{l}
R = \ifthen{\varphi}{\beta\sdot\ifthen{\varphi}{\beta.\nil}} \qquad
S = \ifthen{\varphi}{\beta.\beta.\nil}
\end{array}\]
In $R$ the condition $\varphi$ is checked twice. In general $R$ and $S$ are not equivalent. To see this, let $\Psi$ and $\Psi'$ be such that $\Psi \vdash \varphi$ and $\Psi \ftimes \Psi' \not\vdash \varphi$. We then have that
$\framedtransempty{\Psi}{R}{\beta}{\ifthen{\varphi}{\beta.\nil}}$ and it cannot be simulated by
$\framedtransempty{\Psi}{S}{\beta}{\beta.\nil}$ because of the recurring clause of extension of arbitrary assertion: $\ifthen{\varphi}{\beta.\nil}$ has no transition in the environment $\Psi \ftimes \Psi'$. However,
if the entailment relation satisfies weakening, i.e. $\Psi \vdash \varphi \Rightarrow \Psi \ftimes \Psi' \vdash \varphi$, we get the intuitive result that $R$ and $S$ are bisimilar.  
This also demonstrates the inadequacy of the smaller and simpler definition of $\bisim$ as  the largest relation satisfying
\[\text{if }\forall \Psi. \Psi \frames \trans{P}{\beta}{P'} \text{ then } \Psi \frames \trans{Q}{\beta}{Q'}\; \land\; P'\bisim Q'\]
The difference is that here bisimulation recurringly requires to hold for {\em all} assertions, not only for those that are extensions of the ones passed so far.
This would have the unintuitive effect of making $R$ and $S$ in the example above non-bisimilar, even if weakening holds.

If there are inconsistent assertions, i.e. assertions
that entail all conditions, the effect of Clause~3 is very strong: Bisimilar
agents are required to behave the same even if the environment is inconsistent.
For example, in this situation the agent $(\nu a)\overline{a}\sdot \nil$ is not
equivalent to $\nil$, since an inconsistent assertion can make all names channel
equivalent, and therefore $(\nu a)\overline{a}\sdot \nil$ has actions with all
names except $a$ as subject.
The algebraic properties to follow hold for all psi-calculi, including those with inconsistent assertions. It remains to be seen if and how bisimulation in such psi-calculi is useful to model applications.

Interestingly, there is an alternative way to define bisimulation as a {\em binary} relation preserved by parallel contexts.
\begin{definition}[Context bisimulation]
A {\em context bisimulation}
 $\mathcal{R}$ is a binary relation on agents such that
 ${\mathcal{R}}(P,Q)$ implies all of
 \begin{enumerate}[(1)]%
 \item Static equivalence:
  $\fr{P} \sequivalent \fr{Q}$
 \item
   Symmetry: ${\mathcal{R}}(Q,P)$
 \item
 Extension of contextual assertion:
 $\forall \Psi.\ \mathcal{R}(\pass\Psi\parop P,\;\pass\Psi\parop Q)$
 \item   Simulation:
for all $\alpha, P'$ such that
$\bn{\alpha}\freshin Q$ there exists a $Q'$ such
that
\[\text{if }\framedtransempty{\emptyframe}{P}{\alpha}{P'}\quad \text{ then } \quad 
\framedtransempty{\emptyframe}{Q}{\alpha}{Q'}
\;\land\;  {\mathcal R}(P', Q')\]
\end{enumerate}
We define $P \bisim^c Q$ to mean that there exists a context bisimulation ${\mathcal{R}}$ such that
${\mathcal{R}}(P,Q)$.

 \label{def:contextbisim}
\end{definition}
 Such a definition
is more in line with standard contextual bisimulations, and also the way bisimulation is defined in the applied pi-calculus. The drawback is that it
relies on an operator in the calculus (parallel) for its definition. For
conducting proofs our experience is that Definition~\ref{def:bisim} is
preferable. We have shown that these bisimilarities coincide, i.e., the
definitions result in the same bisimulation equivalence:
\begin{theorem}[Bisimilarity and context bisimilarity coincide]
\label{thm:contextbisim}
$\bisim \;= \; \bisim^c$
\end{theorem}

We now show that
the usual strong early bisimilarity for the \pic{}, denoted $\bisim_\pi$, and
bisimilarity in the instance \piinstance{}
coincide.
\begin{theorem}[\pic{} bisimilarity and \piinstance{}
bisimilarity coincide]
\label{thm:picoincide}
 \[\bisimpi{P}{Q} \Leftrightarrow \encodepi{P} \bisim \encodepi{Q}\]
\end{theorem}
\begin{proof}
 $(\Rightarrow)$: Static equivalence and extension of arbitrary
assertions hold trivially since the only assertion is $\unit$. Symmetry follows
directly, and simulation follows from
Lemma~\ref{lemma:correspondence-pi-encoding}.\\
  $(\Leftarrow)$: Symmetry follows directly, and simulation follows from
Lemma~\ref{lemma:correspondence-pi-encoding}.
\end{proof}

In addition, we conjecture that Inside-outside bisimilarity for the pi-F
calculus~\cite[Definition 17]{wischik:thesis} coincides with bisimilarity for
the
psi-calculus \textbf{Fi} (see Section~\ref{sec:ccpi}).

\subsection{Algebraic properties}
\label{sec:properties}
Our results are that
bisimilarity is preserved by the operators in the expected way, and also satisfies the expected structural algebraic laws.
\begin{theorem}
\label{thm:congr} For all $\Psi$:
 \begin{enumerate}[\em(1)]
\setlength{\itemsep}{0pt}
 \item $P \bisim_\Psi Q  \Longrightarrow P \pll R \bisim_\Psi Q \pll R$. \label{parl}
 \item $P \bisim_\Psi Q \Longrightarrow \res{a}P \bisim_\Psi \res{a}Q$ \label{res} 
 \qquad \mbox{if $a \freshin \Psi$.}
 \item $P \bisim_\Psi Q \Longrightarrow \,!P \bisim_\Psi\, !Q$. \label{bang}
 \item  \label{case}{\rm $\forall i.P_i \bisim_\Psi Q_i  \Longrightarrow
        \caseonly{\ci{\ve{\varphi}}{\ve{P}}} \;{\bisim_\Psi}\; \caseonly{\ci{\ve{\varphi}}{\ve{Q}}}$.}
\item $P \bisim_\Psi Q  \Longrightarrow \out{M}N . P \bisim_\Psi \out{M}N.Q$. \label{out}
\item $(\forall \ve{L}.\; P\lsubst{\ve{L}}{\ve{a}} \bisim_\Psi Q\lsubst{\ve{L}}{\ve{a}})$ $\;\Longrightarrow\;$ \\
$\lin{M}{\ve{a}}{N.P}  \bisim_\Psi \lin{M}{\ve{a}}{N.Q} $
\qquad \mbox{if $\ve{a} \freshin \Psi$.}
\end{enumerate}
\end{theorem}

\begin{definition}\label{def:congruence}
$P\sim_\Psi Q$ means that for all sequences $\sigma$ of substitutions 
it holds that
$P\sigma \bisim_\Psi Q\sigma$, and we write $P \sim Q$ for $P \sim_\emptyframe Q$.
\end{definition}

Our requirements on the substitution function are very weak. 
For example, we do not require that $P\lsubst{\epsilon}{\epsilon}$ (the substitution of length 0) is $P$,
nor that sequences of substitutions $\lsubst{\vec{M}}{\vec{x}}\lsubst{\vec{N}}{\vec{y}}$ can be combined into one.
For this reason, $\sim_\Psi$ is defined by closure under \emph{sequences} of substitutions rather than single substitutions $\lsubst{\vec{M}}{\vec{x}}$.

\begin{theorem}\label{thm:fullcong}
$\sim_\Psi$ is a congruence for all $\Psi$.
\end{theorem}
\begin{theorem} $\sim$ satisfies the following structural laws:
\label{thm:struct}
\begin{mathpar}
 \begin{array}{rcll}
  P & \sim& P \pll \nil \\
  P \pll (Q \pll R) & \sim& (P \pll Q) \pll R\\
  P \pll Q & \sim& Q \pll P\\

  (\nu a)\nil & \sim& \nil\\
  P \pll (\nu a)Q & \sim& (\nu a)(P \pll Q) & \text{if }a \freshin P\\
  \out{M}{N}.\res{a}P & \sim& \res{a}\out{M}{N}.P &\text{if }a \freshin M,N\\
  \lin{M}{\ve{x}}{N}.\res{a}P & \sim& \res{a}\lin{M}{\ve{x}}{(N)}.P&
      \text{if }a \freshin \ve{x},M,N\\
  \caseonly{\ci{\ve{\varphi}}{\ve{(\nu a)P}}} & \sim& 
      (\nu a)\caseonly{\ci{\ve{\varphi}}{\ve{P}}}&\text{if }a \freshin \ve{\varphi}\\
  \res{a}\res{b}P & \sim& \res{b}\res{a}P\\
  !P & \sim& P \pll !P \\
 
\end{array}
\end{mathpar}

\end{theorem}

\label{section:proof_examples}
The most awkward part of the proofs is for Theorem~\ref{thm:congr}(\ref{parl}),
and historically this is the proof that most often fails in calculi of this
complexity; the intricate correspondences between parallel processes and their
assertions are hard to get completely right. We give an outline of the proof
and cover in detail the simulation case where the parallel processes
communicate with each other. In the following we tacitly assume $\fr{P} =
\framepair{\frnames{P}}{\frass{P}}$, where $\frnames{P} \freshin P$, for any
agent $P$, unless otherwise noted.

We pick the candidate relation $\R =
\{(\Psi, (\nu \ve{a})(P \pll R),(\nu \ve{a})(Q \pll
R)) : P\bisim_{\Psi \ftimes \frass{R}} Q\}$ where $\ve{a}
\freshin \Psi$, and prove that $\R$ is a bisimulation. Moreover we assume that $\frnames{P}
\freshin \frnames{Q}, Q, \frnames{R}, R, \Psi$, and
$\frnames{R} \freshin P, Q, \Psi$, or, in other
words, that bound names are distinct from all free names and other bound names.
 Formally the proof is conducted by an
induction on the length of $\ve{a}$. The induction step is straightforward,
so we focus on the base case. The agent $P\parop R$ can operate either by $P$ or
$R$ doing individual actions, or by $P$ and $R$ communicating, where we cover
the latter case, as it is the most involved.

In this case we have, by the {\sc Com} rule, that $P$ does an input transition
($\framedtransempty{\Psi \ftimes \frass{R}}{P}{\inn{M}{N}}{P'}$), $R$ does an
output transition ($\framedtransempty{\Psi \ftimes
\frass{P}}{R}{\bout{\ve{a}}{K}{N}}{R'}$), and that the subjects of the
transitions are channel equivalent ($\Psi \ftimes \frass{P} \ftimes \frass{R}
\vdash M \sch K$). The resulting communication between $P$ and $R$ is thus
$\framedtransempty{\Psi}{P\parop R}{\tau}{(\nu\vec{a})(P'\parop R')}$.

To complete this step of the proof we need to find a $Q'$ such that
$\framedtransempty{\Psi}{Q\parop R}{\tau}{(\nu\vec{a})(Q'\parop R')}$, and
$(\Psi,\ (\nu\vec{a})(P'\parop R'),\ (\nu\vec{a})(Q'\parop R'))\in\R$.

The presence of assertions in the transitions complicates the proof. We know
that $P\bisim_\Psi Q$, and hence by
Definition~\ref{def:bisim}(\ref{def:case:bisimExtAss}) that
$P\bisim_{\Psi\ftimes\frass{R}} Q$. Since $\framedtransempty{\Psi \ftimes
\frass{R}}{P}{\inn{M}{N}}{P'}$, we can obtain a $Q'$ such that
$\framedtransempty{\Psi \ftimes \frass{R}}{Q}{\inn{M}{N}}{Q'}$ and
$P'\bisim_{\Psi\ftimes\frass{R}} Q'$. However, this transition cannot
communicate with $\framedtransempty{\Psi \ftimes
\frass{P}}{R}{\bout{\ve{a}}{K}{N}}{R'}$, since that transition is derived by the
assertion $\Psi \ftimes \frass{P}$, and not $\Psi \ftimes \frass{Q}$. Moreover,
$M$ and $K$ are channel equivalent by the assertion $\Psi \ftimes \frass{P}
\ftimes \frass{R}$, and not $\Psi \ftimes \frass{Q} \ftimes \frass{R}$, which
would be needed to derive the desired communication. In order to complete the
proof, we need a lemma which switches the occurrences of $\frass{P}$ to
$\frass{Q}$ in the transition of $R$, as well as in the channel equality.

Once a communication has been derived, we must prove that the corresponding
derivatives $(\nu\vec{a})(P'\parop R')$, and $(\nu\vec{a})(Q'\parop R')$ are in
the candidate relation $\R$. From the definition of $\R$ we get that this holds
if $P'\bisim_{\Psi\ftimes\frass{R'}} Q'$, but we only know that
$P'\bisim_{\Psi\ftimes\frass{R}} Q'$. In order to complete the proof, $P'$ and
$Q'$ must be bisimilar in the assertion $\Psi\ftimes\frass{R'}$, and not only in
$\Psi\ftimes\frass{R}$.

We provide lemmas which will address both of these obstacles in turn, after
which this proof will be concluded. Lemma \ref{lemma:simulate-COM}
simultaneously changes the assertion deriving the transition for $R$, and the
channel equality, and Lemmas \ref{lemma:transition-compensation} and
\ref{lemma:statEqBisim} ensure that the derivatives of the communicating agents
are in the candidate relation $\R$. Lemmas \ref{lemma:find-subject} and
\ref{lemma:subject-rewriting} are two generally applicable lemmas used to prove
Lemma \ref{lemma:simulate-COM}. We define $\subj{\bout{\ve{a}}{M}{N}} = M$ and
similarly for input actions.

The first lemma shows that given a finite set of names $B$ that are fresh for
$P$ we can
find a term $M$ channel equivalent to the subject of an action from $P$ whose
names are fresh for $B$.

\begin{lemma}[Find equivalent term]
\label{lemma:find-subject}
 \[\begin{array}{rl}
         & B \subseteq \mathcal{N} \land B\text{ finite } \land B \freshin P\\
   \land & \framedtransempty{\Psi}{P}{\alpha}{P'}\text{ where }\alpha\neq\tau\\
   \land & \frnames{P} \freshin \Psi, P, \subj{\alpha}, B \\
   \Longrightarrow & \exists M\quad . \quad B \freshin M\\
   & \quad \quad \quad \land \; \Psi \ftimes \frass{P}
\vdash M \sch \subj{\alpha}
  \end{array}\]
\end{lemma}
\begin{proof}
 A straightforward induction on the length of the derivation of the
transition. In the base case we
choose $M$ as the prefix in the agent.
\end{proof}

The next lemma shows that given a transition we can find another transition
whose subject is channel equivalent to the subject of the original transition
and that leads to the same derivative as the original transition.

\begin{lemma}[Rewrite subject]
\label{lemma:subject-rewriting}
 \[\begin{array}{rl}
      & \framedtransempty{\Psi}{P}{\bout{\ve{a}}{M}{N}}{P'}\\
   \land & \Psi \ftimes \frass{P} \vdash K \sch M\\
   \land & \frnames{P} \freshin \Psi, P, K, M\\
   \Longrightarrow & \framedtransempty{\Psi}{P}{\bout{\ve{a}}{K}{N}}{P'}
   \end{array}\]
The symmetric lemma where $P$ does an input is omitted.
\end{lemma}
\begin{proof}
 A straightforward induction on the length of the derivation of the
transition.
\end{proof}

We can now prove the lemma which allows us to simultaneously switch the assertions deriving a transition, as well as channel equality in a communication. This lemma looks a bit intimidating and the proof details can safely be skipped at a first
reading. It says that if $P$ and $Q$ are bisimilar and $P$ can communicate with
$R$ via the channel $K$, then there exists a channel $K'$ such that $Q$ can
communicate with $R$ via $K'$.

\begin{lemma}[Switching]
  \label{lemma:simulate-COM}
\[
 \begin{array}{rl} 
\ & P\bisim_{\Psi\ftimes\frass{R}} Q \\
  \wedge & \framedtransempty{\Psi \ftimes \frass{R}}{P}{\inn{M}{N}}{P'}\\
  \land & \framedtransempty{\Psi \ftimes \frass{P}}{R}{\bout{\ve{a}}{K}{N}}{R'}\\
  \land & \Psi \ftimes \frass{P} \ftimes \frass{R} \vdash K \sch M\\

  \land & \frnames{R} \freshin \frnames{P}, \frnames{Q}, \Psi, P, Q, R, K \\
  \land & \frnames{Q} \freshin \Psi, R, P, Q, M\\
  \land & \frnames{P} \freshin R,M,\Psi\\
  \Longrightarrow & \exists K' . \framedtransempty{\Psi \ftimes
\Psi_Q}{R}{\bout{\ve{a}}{K'}{N}}{R'}\\
  \land & \Psi \ftimes \Psi_Q \ftimes \Psi_R \vdash K' \sch M\\
  \land & \frnames{R} \freshin K'
 \end{array}
\]
There is also a symmetric lemma where $R$ does an input.
\end{lemma}
\proof
 By induction on the length of the derivation of the transition from $R$. We only look
at one base case and one induction step here. The other cases are similar.
\begin{enumerate}[\hbox to8 pt{\hfill}]
 \item\noindent{\hskip-12 pt\textsc{Out}:}\ In this case $R = \out{K_s}{N}.R'$ for some term 
$K_s$, and the transition is derived like this:
\[
\inferrule*[left=\textsc{Out}]
    {\Psi \ftimes \frass{P} \vdash K_s \sch K }
    {\framedtransempty{\Psi \ftimes
\frass{P}}{\out{K_s}{N}.R'}{\out{K}{N}}{R'}}
\]
Since $\frnames{P} \freshin \Psi, R$ we get that $\Psi \ftimes \fr{P}
\vdash
K_s \sch K_s$. This in turn gives us that $\Psi \ftimes \fr{Q} \vdash K_s \sch
K_s$, which means that $\Psi \ftimes \frass{Q}
\vdash K_s \sch K_s$. We then establish the first conjunct by:
\[
\inferrule*[left=\textsc{Out}]
    {\Psi \ftimes \frass{Q} \vdash K_s \sch K_s }
    {\framedtransempty{\Psi \ftimes
\frass{Q}}{\out{K_s}{N}.R'}{\out{K_s}{N}}{R'}}
\]
For the second conjunct, we have that $\Psi
\ftimes \frass{P} \vdash K_s \sch K$ and that $\Psi \ftimes \frass{P}
\ftimes \emptyframe \vdash K \sch M$ (since in this case $\frass{R}$ is
$\emptyframe$).
Identity and transitivity then give us that $\Psi \ftimes \frass{P}
\vdash K_s
\sch M$. Since $\frnames{P} \freshin R, M$ we have that $\Psi \ftimes \fr{P}
\vdash K_s \sch M$ and since $P$ and $Q$ are bisimilar we also have that $\Psi
\ftimes \fr{Q} \vdash K_s \sch M$. We finally get $\Psi \ftimes \frass{Q}
\vdash K_s \sch M$. The third conjunct is trivial since  $\frnames{R}$ is empty. 

 \item\noindent{\hskip-12 pt\textsc{Scope}:}\ In this case $R = (\nu b)R'$ for some name $b$ and the
transition is derived like this:
\[
\inferrule*[Left=\textsc{Scope}, right={$b \freshin \bout{\ve{a}}{K}{N},\Psi$}]
    {\framedtransempty{\Psi \ftimes \frass{P}}{R'}{\bout{\ve{a}}{K}{N}}{R''}}
    {\framedtransempty{\Psi \ftimes \frass{P}}{(\nu
b)R'}{\bout{\ve{a}}{K}{N}}{(\nu b)R''}}
\]
Let $b \freshin \frnames{P}, \frnames{Q}, P, Q$.
Note that by definition we have $\frass{(\nu b)R'} = \frass{R'}$. We also
have that \vfill\eject

\noindent$\frnames{(\nu b)R'} \freshin \frnames{P}, \frnames{Q}, \Psi, P, Q,
(\nu b)R', K
\Longrightarrow \frnames{R'} \freshin \frnames{P}, \frnames{Q}, \Psi, P,
Q, R', K$ and that $\frnames{(\nu b)R'} \freshin \frnames{P}, \frnames{Q} \land
\frnames{P}, \frnames{Q} \freshin (\nu b)R' \land b \freshin
\frnames{P}, \frnames{Q} \Longrightarrow \frnames{P},
\frnames{Q} \freshin R'$. From the induction hypothesis we then get that
$\framedtransempty{\Psi \ftimes \frass{Q}}{R'}{\bout{\ve{a}}{K'}{N}}{R''}$,
$\Psi \ftimes \frass{Q} \ftimes \frass{R'} \vdash M \sch K'$, and that
$\frnames{R'} \freshin K'$.

From the fact that $P$ and $Q$ are bisimilar we get that
$\framedtransempty{\Psi \ftimes \frass{(\nu b)R'}}{Q}{\inn{M}{N}}{Q'}$. Let $B
= \{b\} \cup \frnames{R'}$. By
Lemma \ref{lemma:find-subject} we learn that there exists a term $K''$ such
that $\Psi \ftimes \frass{(\nu b)R'} \ftimes \frass{Q} \vdash K'' \sch M$,
fulfilling the second obligation, and
that $B \freshin K''$. This gives us that
$\frnames{R'}, b \freshin K''$. By transitivity we then get that $\Psi \ftimes
\frass{(\nu b)R'} \ftimes \frass{Q} \vdash K' \sch K''$. We now use Lemma
\ref{lemma:subject-rewriting} to get that $\framedtransempty{\Psi \ftimes
\frass{Q}}{R'}{\bout{\ve{a}}{K''}{N}}{R''}$. Finally we do the following
derivation:
\[
\inferrule*[Left=\textsc{Scope}, right={$b \freshin
\bout{\ve{a}}{K''}{N},\Psi$}]
    {\framedtransempty{\Psi \ftimes \frass{Q}}{R'}{\bout{\ve{a}}{K''}{N}}{R''}}
    {\framedtransempty{\Psi \ftimes \frass{Q}}{(\nu
b)R'}{\bout{\ve{a}}{K''}{N}}{(\nu b)R''}}
\]
That $\frnames{(\nu b)R'} \freshin K''$ follows from $B \freshin K''$.\qed
\end{enumerate}

\noindent The following lemma proves that when an agent performs a transition, its
frame is extended with a new assertion ($\Psi'$ below):

\begin{lemma}
\label{lemma:transition-compensation}
If $\framedtransempty{\Psi}{R}{\inn{M}{N}}{R'}$ and $\frnames{R} \freshin
R,N,C$ where $C$ is a set of names, then $\exists \Psi',
\frnames{R'},\frass{R'}$ such that $\fr{R'}
= \framepair{\frnames{R'}}{\frass{R'}}\;\land\;
\frass{R} \ftimes \Psi' \sequivalent \frass{R'}\;\land\; \frnames{R'} \freshin
C,R'$.
\end{lemma}
\proof
 A straightforward induction on the length of the derivation of the
 transition.\qed

Finally, we need a lemma which allows us to switch the environment for a bisimulation for an equivalent one.
\begin{lemma}
If $\Psi\frames P\bisim Q$ and $\Psi\sequivalent \Psi'$ then $\Psi'\frames P\bisim Q$
\label{lemma:statEqBisim}
\end{lemma}
\proof
The candidate relation for the bisimulation is $\mathcal{R} = \{(\Psi',\ P,\ Q) : \Psi\frames P\bisim Q\wedge\Psi\sequivalent\Psi'\}$.
The four cases are proved separately.
\begin{enumerate}[\hbox to8 pt{\hfill}]
\item\noindent{\hskip-12 pt\bf Case 1:}\ Follows from the fact that $\ftimes$ is compositional, where the bound names of the frames of $P$ and $Q$ are alpha-converted not to clash with $\Psi'$.
\item\noindent{\hskip-12 pt\bf Case 2:}\ $\mathcal{S}$ is trivially symmetric, since $\bisim$ and $\sequivalent$ are symmetric.
\item\noindent{\hskip-12 pt\bf Case 3:}\ Follows from the fact that $\ftimes$ is compositional.
\item\noindent{\hskip-12 pt\bf Case 4:}\ From the definition of $\bisim$ and the transition $\framedtransempty{\Psi}{P}{\alpha}{P'}$, we obtain a $Q'$. s.t. $\framedtransempty{\Psi}{Q}{\alpha}{Q'}$ and $\Psi\frames P'\bisim Q'$. By induction on the derivation of this transition, and the fact that $\Psi\sequivalent \Psi'$, we get that $\framedtransempty{\Psi'}{Q}{\alpha}{Q'}$. Moreover, since $\Psi\frames P'\bisim Q'$ and $\Psi\sequivalent \Psi'$ we have that $(\Psi',\ P', \ Q')\in\mathcal{S}$.\qed
\end{enumerate}\smallskip

\noindent With these lemmas in place we complete the proof of
Theorem~\ref{thm:congr}(\ref{parl}) commenced at the beginning of this section.
The case we are proving is when $P\pll R$ performs a communication. We must find
a corresponding transition from $Q\parop R$ such that the derivatives remain in
the candidate relation $\R$. The agents $P$ and $R$ can communicate using the
following derivation.
\[
\inferrule*[left=\textsc{Com}, right={$\inferrule{}{\ve{a}
\freshin R}$}]
 {\framedtransempty{\Psi_R \ftimes \Psi}{P}{\bout{\ve{a}}{M}{N}}{P'} \\
  \framedtransempty{\Psi_P \ftimes \Psi}{R}{\inn{K}{N}}{R'} \\
  \Psi \ftimes \Psi_P \ftimes \Psi_R \vdash M \sch K
  }
       {\framedtransempty{\Psi}{P \pll R}{\tau}{(\nu \ve{a})(P' \pll R')}}
\]
Our goal is to replace $P$ with $Q$ in the premises so that we can derive
the simulating transition. Let $\fr{Q} =  \framepair{\frnames{Q}}{\frass{Q}}$ be
such that $\frnames{Q} \freshin P, \frnames{R}, R, \Psi, M$.

We  use Lemma \ref{lemma:simulate-COM}, to obtain
$\framedtransempty{\Psi_Q \ftimes \Psi}{R}{\inn{K'}{N}}{R'}$ and
$\Psi \ftimes \Psi_Q \ftimes \Psi_R \vdash M \sch
K'$. Since $P$ and $Q$ are bisimilar we
have that $\framedtransempty{\Psi \ftimes
\frass{R}}{Q}{\bout{\ve{a}}{M}{N}}{Q'}$.
We then derive the following:
\[
\inferrule*[left=\textsc{Com}, right={$\inferrule{}{\ve{a}
\freshin R}$}]
 {\framedtransempty{\Psi_R \ftimes \Psi}{Q}{\bout{\ve{a}}{M}{N}}{Q'} \\
  \framedtransempty{\Psi_Q \ftimes \Psi}{R}{\inn{K'}{N}}{R'} \\
  \Psi \ftimes \Psi_Q \ftimes \Psi_R \vdash M \sch K'
  }
       {\framedtransempty{\Psi}{Q \pll R}{\tau}{(\nu \ve{a})(Q' \pll R')}}
\]\smallskip

\noindent We know that $P'\bisim_{\Psi \ftimes \frass{R}} Q'$ and by clause \ref{def:case:bisimExtAss} in the definition of bisimulation (extension of arbitrary assertion) that $P'\bisim_{\Psi \ftimes \frass{R} \ftimes \Psi'} Q'$ for any $\Psi'$. By Lemma \ref{lemma:transition-compensation}
we know that there exists a $\Psi''$ such that $\frass{R} \ftimes \Psi''
\sequivalent \frass{R'}$, so in particular, using Lemma \ref{lemma:statEqBisim}, we have that $P'\bisim_{\Psi\ftimes\frass{R'}} Q'$
We then conclude that $(\Psi, (\nu \ve{a})(P' \pll R'),(\nu \ve{a})(Q' \pll R'))
\in\R$.\qed

The proofs of theorems~\ref{thm:contextbisim},\ref{thm:congr}--\ref{thm:struct}
follow a similar pattern, using induction over the lengths of
the derivations of the transitions. The part we have just shown is the
most challenging. Further proofs are found in \cite{magnus.johansson:thesis}.

\section{Formalisation in Isabelle}\label{sect:formalisation}

As the complexities of calculi increase, the proofs become more complicated and therefore more error prone. In Section~\ref{sec:expressiveness} we discussed how both the applied {\pic} and the concurrent constraint {\pic} have turned out to be non-compositional. 
This hints at the complexity of the proofs and the difficulty of getting them right. 
Our proofs for psi-calculi are also sometimes long and intricate. For example,
the proof sketch of Theorem~\ref{thm:congr}(\ref{parl}), described in the previous section, is substantially more complicated than its corresponding proof for the {\pic}. However, we emphasise that the proof is not substantially different in structure: it is just a set of properties of transitions, all established by induction over the the definition of the semantics. In this, psi-calculi are simpler than many other calculi that rely on stratified definitions of the semantics with devices such as a structural congruence.

In order to ensure that proofs are correct, automated and interactive proof assistants or theorem provers can be used to formally verify the proofs with the aid of a computer. 
 We have completely formalised all results in Section~\ref{sec:bisimilarity},
with the exception of Theorem~\ref{thm:picoincide}, in the interactive theorem
prover Isabelle. %
To the best of our knowledge, no
calculus of this complexity has previously been formalised in a theorem prover. We have earlier~\cite{bengtsonparrow2007} formalised a substantial part of the
{\pic} meta-theory in Isabelle. This section will cover the main extensions needed to
formalise the framework for psi-calculi. More in-depth expositions are found
in \cite{BengtsonParrowTPHOLs2009,jesper.bengtson:thesis}.

\subsection{Alpha-equivalence}

The main difficulty with formalising any process algebra in a theorem prover is to reason about alpha-equivalence in a convenient way. When conducting manual proofs on paper this notion is often glossed over, and statements such as ``we assume any bound name under consideration to be sufficiently fresh'' are commonplace. For machine checked proofs this poses a problem. Exactly what does it mean for a bound name to be sufficiently fresh?

We use Nominal Isabelle \cite{U07:NominalTechniquesInIsabelleHOL} to formalise
datatypes with binders, and to reason about them up to alpha-equivalence; in
other words, all our proofs deal with alpha equivalence classes rather than with
particular representatives. 
As usual
alpha variants of agents are identified , so e.g. $(\nu a)P = (\nu b)((a\;b)\cdot P)$, when $b\freshin P$, and similarly for names bound in the input construct. 
Formally, name swapping on agents distributes over all constructors, and substitution on agents avoids captures by binders through alpha-conversion as usual. In that way Nominal Isabelle provides an alpha-equivalence class of agents where
the support of $P$ is the union of
the supports of the components of $P$, removing the bound names. This
corresponds to the names with a free occurrence in $P$. %

Frames contain binders and we reason about their alpha equivalence classes in
the same way.  Also, transitions contain binders. Consider the output transition
$\Psi\frames\trans{P}{\out{M}{(\nu \tilde{a})N}}{P'}$. To be completely formal,
as described in \cite{bengtsonparrow2007}, $\tilde{a}$ is a binding occurrence
with a scope that contains both $N$ and $P'$. We accomplish this by creating a
datatype containing both an action and the derivative process as follows.

\begin{definition}[Residuals]

A {\em residual} with the action $\alpha$ and the derivative $P'$, is written $\alpha\prec P'$. 
\end{definition}
Thus we have the following three forms of residuals:

\[\begin{array}{lllll}

\out{M}(\nu\vec{a}){N} \prec P'& \mbox{Output} \\
\inn{M}{N}                  \prec P'& \mbox{Input}  \\
\tau                      \prec P'& \mbox{Silent}
\end{array}\]
In the Output residual, $\vec{a}$ binds into both $N$ and $P'$.
In this way we get a nominal datatype of residuals where name swapping just
distributes to its components and the support is the free names. A transition is
then simply a pair consisting of an agent and a residual. Again, Nominal
Isabelle allows us to reason about alpha equivalence classes of transitions.
Typically a property of transitions is established by induction, with one case
for each rule. This means that we assume the property of the premise of the
rule, and must establish it for the conclusion. Since we work with alpha
equivalence classes it is enough to establish the property for one
representative of the alpha equivalence class. This formalises the principle
that we may always pick bound names fresh.

Datatypes for agents, frames and transitions in Nominal Isabelle require
sequences of binders, e.g.  in the input prefix and in the output action. 
It is important to reason about arbitrarily long binding sequences as atomic objects, otherwise there would be a constant need for inductive proofs over the length of these sequences. 
Nominal Isabelle only supports single binders, and we have therefore created infrastructure to reason about arbitrarily long binding sequences. 
When alpha-converting a binding sequence, we generate a name permutation $p$ which when applied to the sequence makes it sufficiently fresh. 
The same permutation is then applied to everything under the scope of the binders, for example:
\[
\inprefix{M}{\vec x}{N}.P = \inprefix{M}{p\cdot\vec x}{(p\cdot N)}.(p \cdot P)\quad\quad\mbox{if}\; p \subseteq \vec x\times (p\cdot\vec x)\ \mbox{and}\ (p\cdot\vec x)\ \freshin\ (N,\ P)
\]

\noindent The side condition of this alpha-conversion looks a bit intimidating, but intuitively $p$ swaps members of the original binding sequence to other names such that the resulting binding sequence meets the desired freshness constraints.
This style of alpha-conversion was first introduced by Urban and Berghofer, although to the best of our knowledge it is still unpublished. We cover it more extensively in \cite{BengtsonParrowTPHOLs2009}.

\subsection{Formalising parametric calculi}

The framework for psi-calculi is a parametric formalism. A psi-calculus agent
consists of terms, assertions and conditions. This is modelled in Isabelle by
creating a polymorphic datatype with three type variables. A psi-calculus agent
will thus have the type $(\alpha,\ \beta,\ \gamma)\ \mathtt{psi}$, where
$\alpha$, $\beta$, and $\gamma$ represents terms, assertions, and conditions
respectively. All members of these types need to have finite support.

Isabelle has excellent facilities for a parametric style of reasoning through the use of locales \cite{ballarin:types03}. Locales allow us to specify which functions must exist for the parameters, and which assumptions must hold on them. The entire proof structure of the meta theory is then built using the provided locale parameters. When creating a psi-calculus instance, the functions must be provided and the assumptions must be proved. Once this is done, all meta-theoretical proofs will be guaranteed to hold for the new instance.

One requirement from Section~\ref{sec:nominal} is that there is a substitution
function which substitutes terms for names in assertions, conditions and terms.
To this end, a locale is created with a substitution function of type
$\delta\to\mbox{\tt name list}\to\alpha\ \mathtt{list}\to\delta$, where the type
$\alpha$ will be what we use for terms, and the type $\delta$ can be any of the
three nominal sets. The locale contains the following assumptions, which
implement the requirements of a substitution function mentioned in
Section~\ref{sec:nominal}

\begin{center}
\begin{tabular}{@{}ll}
Equivariance: & $p\cdot(X[\vec{x}:=\vec{T}]) = (p\cdot X)\big[(p\cdot\vec{x}):=(p\cdot\vec{T})\big]$\\
Freshness: 
\; &if $\; \vec{x}\subseteq\n(X)$ and $a\freshin X[\vec{x}:=\vec{T}]$ then $a\freshin\vec{T}$\\
Alpha-equivalence: & if $p\subseteq \vec{x}\times(p\cdot \vec{x})$ and $(p\cdot\vec{x})\freshin X$ then\\
\; & \qquad $X[\vec{x}:=\vec{T}] = (p\cdot X)[(p\cdot\vec{x}):=\vec{T}]$
\end{tabular}
\end{center}

The assumptions on this locale are straightforward. As all functions in any
nominal formalisation, substitution must be equivariant.  Freshness is a reformulation of requirement~1 in Section~\ref{sec:nominal}.  Similarly, Alpha-equivalence is requirement~2. Intuitively this means that the vector being substituted is
switched to one which is sufficiently fresh. As an example of its use, consider the {\sc Input} rule.
\[
\inferrule*[Left=\textsc{In}]
    {\Psi \vdash M \sch K }
    {\framedtransempty{\Psi}{\lin{M}{\ve{y}}{N}.P}{\inn{K}{N}\lsubst{\ve{L}}{\ve{y}}}{P\lsubst{\ve{L}}{\ve{y}}}}
\]

\noindent If a proof requires the input agent to be alpha-converted to $\lin{M}{p\cdot\ve{y}}{(p\cdot N)}.(p\cdot P)$ such that $p\cdot\vec{y}$ is sufficiently fresh, it is necessary to convert $N\lsubst{\ve{L}}{\ve{y}}$ to $(p\cdot N)\lsubst{\ve{L}}{(p\cdot\ve{y})}$, and $P\lsubst{\ve{L}}{\ve{y}}$ to $(p\cdot P)\lsubst{\ve{L}}{(p\cdot\ve{y})}$ to still be able to derive the input transition. The last constraint accomplishes this. This locale is then instantiated three times: for terms, assertions and conditions respectively.

The nominal morphisms in Definition~\ref{def:parameters} are modelled in a locale  which specifies their existence and equivariance properties. Inside this locale we also define equivalence for assertions and frames and provide an infrastructure for reasoning about equivalence.
This locale is then extended with the requisites  in Definition~\ref{def:entailmentrelation}.

Finally, the substitution locale is combined with the locale for equivalence to form an environment in which the rest of the theories can be proved. The locales offer a very intuitive way of reasoning about parametric systems, and without them this formalisation would have been very hard.

\subsection{Encoding partial operators}

In Definition~\ref{def:agents}, there is a well-formedness condition that all agents occurring under a Case or Replication operator must be guarded. Formally, this means that these operators are not total. For example, $\pass\Psi$ is an agent but $!\pass\Psi$ is not.

To represent this in Isabelle, we take the technically easiest approach to augment the {\textsc {Case}} and the {\textsc {Rep}}-rules of the operational semantic with a premise that the agents they operate on are guarded. In effect this allows non well-formed agents, but they have no transitions and are all bisimilar to $\nil$. All Isabelle proofs hold for all agents, so in particular they hold for all well-formed agents. Therefore the Isabelle formalisation establishes the theorems presented in this paper. A few lemmas, for example  that bisimilarity is preserved by Replication, need an extra premise that the agents are guarded, but in the vast majority of lemmas  the necessary properties follow from the operational semantics.

An alternative would be to constrain the datatype representing agents to well
formed agents and thus ensure that all inhabitants of that type meet the
required constraints. This more closely resembles Definition~\ref{def:agents},
and would be the method of choice for use with a theorem prover such as Coq that
supports dependent typing. There the well-formedness conditions can be
integrated into the psi-datatypes, i.e.\ for all proofs we can assume that we
are only dealing with well-formed agents. The downside of this approach is that
whenever an agent is constructed, a proof that it is well-formed must also be
supplied.

A third option to encode partial operators would be to decorate all lemmas which use the well formedness property with an assumption that the agents are well formed. We avoided this since it would clutter up a significant amount of lemmas with extra premises.
\subsection{Results and experiences}

Using Isabelle to formalise the proofs for psi-calculi in parallel to its development has turned out to be invaluable, and we would certainly not have finished successfully without it.
Throughout the development we have uncountable times stumbled over slightly incorrect definitions and not quite correct lemmas, prompting frequent changes in the framework. For example, our mistake in \cite{johansson.parrow.ea:extended-pi} mentioned in Section~\ref{sec:illustrativeexamples} was found during proof mechanisation and would probably not have been found at all without it; at that time we had completed a manual ``proof" that turned out incorrect. The Isabelle formalisation gives us a high degree of confidence in the proved theorems, and equally important, it gives us a repository of proofs and proof strategies that can be re-used when some detail needs to change. Finding out which ramifications a change has on the proofs is quick and straight forward. With manual proofs, changing a detail would mean the boring and dangerously error prone process of going over each proof by hand.

As just one example, in a previous version, the {\sc Case} rule looked as follows:
\[
\inferrule*[left={\textsc{Old-Case}}]
    {\Psi \vdash \varphi_i}
    {\framedtransempty{\Psi}{\caseonly{\ci{\ve{\varphi}}{\ve{P}}}}{\tau}{P_i}}
\]

\noindent
In this rule, the choice of which branch to take in a {\bf case} statement yields an internal action, after which the process $P$ evaluates as usual. An implication is that the requirement that $P$ is guarded can be omitted. We initially adopted this rule since it admits simpler induction proofs. At a quite late stage we decided to change it to the present rule, since this more closely resembles what is used in similar calculi. The change prompted a rework of the entire proof tree from the semantics and up. The total effort was approximately eight hours, and we now know that the new rule does not cause any problems.

Currently we have formally proved
theorems~\ref{thm:contextbisim},\ref{thm:congr}--\ref{thm:struct} using Isabelle,
including all supporting lemmas.
The entire implementation in Isabelle is about 18000 lines of code. It
includes infrastructure for smooth treatment of binding sequences, and it has
developed gradually over two years. The total effort for the present framework
is hard to assess, since it has followed us through many failed attempts and
false starts. Once in place the marginal effort of formalising more results is
manageable. As an example, the total effort in proving
Theorem~\ref{thm:contextbisim}, which was one of the last things we implemented,
was less than one day.

\section{Conclusion and Further Work}
\label{sec:further}
We have defined a framework for mobile process calculi, parametrised on nominal
types for data terms and for a logic to express assertions and conditions. The  expressiveness surpasses the most advanced competing calculi. The semantics is a single inductive definition, which means that proofs are comparatively easy. We have fully formalised the framework in the interactive theorem prover Isabelle, which gives us full confidence in our results on bisimulation and provides a readily available infrastructure for conducting proofs of many instances and variants.

In~\cite{johansson.victor.ea:fully-abstract} we develop a symbolic semantics and bisimulation equivalence, and prove full abstraction with regards to~$\sim$.  This kind of semantics is essential for reducing the state space explosion when exploring transitions and comparing for equivalence, making it ideal for use in automated tools. In~\cite{johansson.bengtson.ea:weak-equivalences} we explore weak bisimulation equivalence, where
$\tau$ actions are considered unobservable. Our results indicate
that the presence of assertions significantly complicates the definitions, in
contrast to the situation with strong bisimulation. Interestingly, for
psi-calculi that satisfy weakening (i.e. $\Psi \vdash \varphi \Longrightarrow
\Psi \ftimes \Psi' \vdash \varphi$) the definitions can be greatly simplified.
We also investigate a barbed equivalence and determine what kind of observations
are needed for full abstraction. The current development of psi-calculi is covered in~\cite{magnus.johansson:thesis} and the associated formalisation in Isabelle is accounted for in~\cite{jesper.bengtson:thesis}.

We intend to explore typed psi-calculi. One idea is to find out what properties the type system
must have in order for the usual theorems such as subject reduction to hold. We are also considering variants of psi-calculi with broadcast communication, where one sender may communicate directly with several receivers, and higher order communication, where agent definitions can be transmitted and executed by the recipient. It seems that both these variants can be accommodated with very small changes of the semantics and that large parts of our formal proofs carry over.

\bibliographystyle{alpha}
\bibliography{bibliography,pi}

\end{document}